%% file: main.tex
\RequirePackage{amsmath}
\documentclass[letterpaper,11pt]{article}
\usepackage{amsthm}
\usepackage{amsmath,amssymb, mathtools}
\usepackage{framed,caption,color,url}
\usepackage{standalone}
\usepackage{fullpage}
\usepackage{mathtools}
\usepackage{algorithm2e,algorithmicx,algpseudocode,tikz,wrapfig}
\usepackage{graphicx}                
\usepackage{cite} 
\usepackage{nicefrac} 
\usepackage{paralist} 
\usepackage{boxedminipage} 
\usepackage[pdfstartview=FitH,colorlinks,linkcolor=blue,filecolor=blue,citecolor=blue,urlcolor=blue,pagebackref=true]{hyperref}
\usepackage{upgreek}

\newif\iffull
\fulltrue

\newif\iflong
\longfalse
\input{macros}

\makeatletter
\newcommand\blfootnote[1]{%
  \begingroup
  \renewcommand\thefootnote{}\footnote{#1}%
  \addtocounter{footnote}{-1}%
  \endgroup
}
\makeatother

\title{New Algorithmic Directions in Optimal Transport and \\ Applications for Product Spaces\blfootnote{\scriptsize A preliminary version of this work appears in The International Symposium on Algorithms and Computation ISAAC 2025.}}

\author{\addtocounter{footnote}{1} Salman Beigi\footnote{\url{salman.beigi@gmail.com}} 
\and Omid Etesami\footnote{\url{omid.etesami@sabanciuniv.edu}} \and Mohammad Mahmoody\footnote{\url{mahmoody@gmail.com}} \and Amir Najafi\footnote{\url{amir.najafi@sharif.edu}}}
\date{}

\begin{document}

\maketitle 

\input{shortAbstract}
 

\clearpage
\tableofcontents
\input{intro}

\input{Definitions}

\input{BasicTools}
\input{Results}

\input{CoM}

\bibliographystyle{plain}

\phantomsection
\addcontentsline{toc}{section}{References}
\bibliography{refs}

\iffull
\clearpage
\appendix
\section*{Appendix}

\input{TechTools}

\input{Reductions}

\section{Borrowed Results}
\input{onedim}
\input{transport-entropy}

\input{empirical}
\fi

 
\end{document}

%% file: macros.tex
\usepackage{centernot}
\usepackage{dsfont}
\usepackage{nicefrac} 

\DeclareSymbolFont{matha}{OML}{txmi}{m}{it}
\DeclareMathSymbol{\push}{\mathord}{matha}{93}

\newcommand{\fDoob}[2]{\ifthenelse{\equal{#2}{}}{\bar{f}(#1)}{\ifthenelse{\equal{#2}{0}}{\bar{f}(\es)}{\bar{f}(#1_{\leq #2})}}}
\newcommand{\fAppr}[2]{\ifthenelse{\equal{#2}{}}{\tilde{f}(#1)}{\ifthenelse{\equal{#2}{0}}{\tilde{f}(\emptyset)}{\tilde{f}(#1_{\leq #2})}}}
\newcommand{\fApprstar}[2]{\ifthenelse{\equal{#2}{}}{\tilde{f^*}(#1)}{\ifthenelse{\equal{#2}{0}}{\tilde{f^*}(\emptyset)}{\tilde{f^*}(#1_{\leq #2})}}}
\newcommand{\ghat}[2]{\ifthenelse{\equal{#2}{}}{\hat{g}(#1)}{\ifthenelse{\equal{#2}{0}}{\hat{g}(\emptyset)}{\hat{g}(#1_{\leq #2})}}}
\newcommand{\mfix}[2]{\ifthenelse{\equal{#2}{}}{m(#1)}{\ifthenelse{\equal{#2}{0}}{m(\emptyset)}{m(#1_{\leq #2})}}}

\newcommand{\fapp}[2]{\ifthenelse{\equal{#2}{}}{\tilde{f}(#1)}{\ifthenelse{\equal{#2}{0}}{\tilde{f}(\emptyset)}{\tilde{f}(#1_{\leq #2})}}}

\usepackage[utf8]{inputenc}

\iffull
\newcommand{\parag}[1]{\paragraph{#1}}

\fi

\renewcommand{\H}{\mathsf{H}}

\newcommand{\TV}{\mathsf{TV}}

\newcommand{\OfIsoSymb}{\mathsf{OffIso}}
\newcommand{\OnIsoSymb}{\mathsf{OnIso}}
\newcommand{\OfExpSymb}{\mathsf{OffExp}}
\newcommand{\OnExpSymb}{\mathsf{OnExp}}
\newcommand{\RecSymb}{\mathsf{Rec}}
\newcommand{\uRecSymb}{\mathsf{u}}
\newcommand{\vRecSymb}{\mathsf{v}}
\newcommand{\LowerSymb}{\mathsf{\ell}}
\newcommand{\BallPrSymb}{\mathsf{\beta}}
\newcommand{\BallSizeSymb}{s}

\newcommand{\ThresholdSymb}{\tau}
\newcommand{\OfRatioSymb}{\mathsf{OffRt}}
\newcommand{\OnRatioSymb}{\mathsf{OnRt}}

\newcommand{\OnRatio}[1][]{\ifthenelse{\equal{#1}{}}{\OnRatioSymb}{\OnRatioSymb^{(#1)}}} 
\newcommand{\OfRatio}[1][]{\ifthenelse{\equal{#1}{}}{\OfRatioSymb}{\OfRatioSymb^{(#1)}}} 

\newcommand{\OfIso}[1][]{\ifthenelse{\equal{#1}{}}{\OfIsoSymb}{\OfIsoSymb^{(#1)}}} 
\newcommand{\OnIso}[1][]{\ifthenelse{\equal{#1}{}}{\OnIsoSymb}{\OnIsoSymb^{(#1)}}} 
\newcommand{\OfExp}[1][]{\ifthenelse{\equal{#1}{}}{\OfExpSymb}{\OfExpSymb^{(#1)}}} 
\newcommand{\OnExp}[1][]{\ifthenelse{\equal{#1}{}}{\OnExpSymb}{\OnExpSymb^{(#1)}}}  

\newcommand{\Rec}[1][]{\ifthenelse{\equal{#1}{}}{\RecSymb}{\RecSymb^{(#1)}}} 

\newcommand{\vRec}[1][]{\ifthenelse{\equal{#1}{}}{\vRecSymb}{\vRecSymb^{(#1)}}} 
\newcommand{\uRec}[1][]{\ifthenelse{\equal{#1}{}}{\uRecSymb}{\uRecSymb^{(#1)}}} 

\newcommand{\Lower}[1][]{\ifthenelse{\equal{#1}{}}{\LowerSymb}{\LowerSymb^{(#1)}}} 

\newcommand{\BallPr}[1][]{\ifthenelse{\equal{#1}{}}{\BallPrSymb}{\BallPrSymb^{(#1)}}} 
\newcommand{\BallSize}[1][]{\ifthenelse{\equal{#1}{}}{\BallSizeSymb}{\BallSizeSymb^{(#1)}}} 
\newcommand{\Threshold}[1][]{\ifthenelse{\equal{#1}{}}{\ThresholdSymb}{\ThresholdSymb^{(#1)}}}

\newcommand{\T}{\mathsf{T}} 
\newcommand{\W}{\mathsf{T}}

\newcommand{\OnC}{\mathrm{OnC}}  
\newcommand{\OnT}{\mathrm{OnT}} 
\newcommand{\OnG}{\mathrm{G}}

\newcommand{\cost}{\mathsf{c}}

\newcommand{\oracleC}{\mathsf{sc}}

\newcommand{\empT}{\T^\mathrm{Em}}

\newcommand{\KL}[2]{\mathsf{KL}(#1,#2)} 

\newcommand{\aSF}{\mathsf{a}}
\newcommand{\gSF}{\mathsf{g}}
\newcommand{\hSF}{\mathsf{h}}

\newcommand{\avr}[2]{\ifthenelse{\equal{#2}{}}{\aSF({#1})}{\ifthenelse{\equal{#2}{0}}{\aSF(\emptyset)}{\aSF({#1}_{\leq #2})}}}

\newcommand{\avrMax}[2]{\ifthenelse{\equal{#2}{}}{\aSF^*({#1})}{\ifthenelse{\equal{#2}{0}}{\aSF^*(\emptyset)}{\aSF^*({#1}_{\leq #2})}}}

\newcommand{\avrApp}[2]{\ifthenelse{\equal{#2}{}}{\tilde{\aSF}({#1})}{\ifthenelse{\equal{#2}{0}}{\tilde{\aSF}(\emptyset)}{\tilde{\aSF}({#1}_{\leq #2})}}}

\newcommand{\avrAppMax}[2]{\ifthenelse{\equal{#2}{}}{\tilde{\aSF}^*({#1})}{\ifthenelse{\equal{#2}{0}}{\tilde{\aSF}^*(\emptyset)}{\tilde{\aSF}^*({#1}_{\leq #2})}}}

\newcommand{\ArgMax}[2]{\ifthenelse{\equal{#2}{}}{\hSF({#1})}{\ifthenelse{\equal{#2}{0}}{\hSF(\emptyset)}{\hSF({#1}_{\leq #2})}}}

\newcommand{\AppArgMax}[2]{\ifthenelse{\equal{#2}{}}{\tilde{\hSF}({#1})}{\ifthenelse{\equal{#2}{0}}{\tilde{\hSF}(\emptyset)}{\tilde{\hSF}({#1}_{\leq #2})}}}

\usepackage{dsfont}

\newcommand{\gain}[2]{\ifthenelse{\equal{#2}{}}{\gSF(#1)}{\gSF(#1_{\leq #2})}}
\newcommand{\gainMax}[2]{\ifthenelse{\equal{#2}{}}{\gSF^*(#1)}{\gSF^*(#1_{\leq #2})}}
\newcommand{\gainApp}[2]{\ifthenelse{\equal{#2}{}}{\tilde{\gSF}(#1)}{\tilde{\gSF}(#1_{\leq #2})}}
\newcommand{\gainAppMax}[2]{\ifthenelse{\equal{#2}{}}{\tilde{\gSF}^*(#1)}{\tilde{\gSF}^*(#1_{\leq #2})}}

\newcommand{\pr}[2][]{\Pr_{\ifthenelse{\isempty{#1}}{}{{#1}}}\left[{#2}\right]}

\usepackage{graphicx}

\newcommand{\tempRemove}[1]{}

\newcommand{\remove}[1]{}

\newcommand{\ol}{\overline}

\newcommand{\nin}{\not \in}

\newcommand{\es}{\varnothing} 
\newcommand{\se}{\subseteq}

\newcommand{\nf}[2]{\nicefrac{#1}{#2}}

\newcommand{\set}[1]{\left\{ #1 \right\}}

\newcommand{\bits}{\{0,1\}}

\newcommand{\norm}[1]{\left\lVert#1\right\rVert}
\newcommand{\To}{\mapsto}

\newcommand{\R}{{\mathbb R}}

\newcommand{\cA}{{\mathcal A}}

\newcommand{\cC}{{\mathcal C}}

\newcommand{\cE}{{\mathcal E}}

\newcommand{\cM}{{\mathcal M}}
\newcommand{\cN}{{\mathcal N}}

\newcommand{\cQ}{{\mathcal Q}}

\newcommand{\cS}{{\mathcal S}}
\newcommand{\cT}{{\mathcal T}}
\newcommand{\cU}{{\mathcal U}}
\newcommand{\cV}{{\mathcal V}}

\newcommand{\cX}{{\mathcal X}}
\newcommand{\cY}{{\mathcal Y}}

\newcommand{\lowerbound}{\Lambda}
\newcommand{\lowerprod}{\Delta}

\usepackage{upgreek}
\usepackage{bm}

\newcommand{\eps}{\varepsilon}

\newcommand{\poly}{\operatorname{poly}}

\newcommand{\Exp}{\operatorname*{\mathbb{E}}}
\newcommand{\Ex}{\Exp}

\newcommand{\Supp}{\operatorname{Supp}}

\makeatletter
\def\cleartheorem#1{%
    \expandafter\let\csname#1\endcsname\relax
    \expandafter\let\csname c@#1\endcsname\relax
}
\makeatother

\theoremstyle{plain}
\iffull
 \newtheorem{theorem}{Theorem}
 \newtheorem{claim}[theorem]{Claim}
 \newtheorem{proposition}[theorem]{Proposition}
 \newtheorem{lemma}[theorem]{Lemma}
 \newtheorem{corollary}[theorem]{Corollary}

\fi

\theoremstyle{definition}

\iffull
 \newtheorem{definition}[theorem]{Definition}

 \newtheorem{remark}[theorem]{Remark}
 
\fi

\theoremstyle{definition}

\newcommand{\sdotfill}{\textcolor[rgb]{0.8,0.8,0.8}{\dotfill}} 

\makeatletter
\def\th@protocol{%
    \normalfont 
    \setbeamercolor{block title example}{bg=orange,fg=white}
    \setbeamercolor{block body example}{bg=orange!20,fg=black}
    \def\inserttheoremblockenv{exampleblock}
  }
\makeatother
\theoremstyle{protocol}

\newtheorem{proto}[theorem]{Protocol}
\newtheorem{protoc}[theorem]{Protocol}

\newcommand{\namedref}[2]{#1~\ref{#2}}

\newcommand{\torestate}[3]{%
\expandafter \def \csname BBRESTATE #2 \endcsname{#3}
\theoremstyle{plain}
\newtheorem{BBRESTATETHMNUM#2}[theorem]{#1}
\begin{BBRESTATETHMNUM#2}\label{#2}\csname BBRESTATE #2 \endcsname   \end{BBRESTATETHMNUM#2}
\newtheorem*{BBRESTATETHMNONNUM#2}{\namedref{#1}{#2}}
}

\newcommand{\restate}[1]{\begin{BBRESTATETHMNONNUM#1}[Restated] \csname BBRESTATE #1 \endcsname
\end{BBRESTATETHMNONNUM#1}}

\newcommand{\src}{{{\mu}}}
\newcommand{\trg}{{{\nu}}}

\newcommand{\Mnote}[1]{}
\newcommand{\Anote}[1]{}
\newcommand{\Onote}[1]{}
\newcommand{\Snote}[1]{}

%% file: shortAbstract.tex
We consider the problem of optimal transport between two high-dimensional distributions $\mu,\nu$ in $\R^n$ from a new algorithmic perspective, in which we are given a sample $x \sim \mu$ and we have to find a close $y \sim \nu$ while running in $\poly(n)$ time, where $n$ is the size/dimension of $x,y$.
n other words, we aim to bound the running time by the dimension of the spaces, rather than by the total size of the representations of the two distributions.
Our main result is   a general algorithmic transport result between any product distribution $\mu$ and an arbitrary distribution $\nu$ of total cost $\Delta + \delta$ under $\ell_p^p$ cost; here $\Delta$ is the cost of the so-called Knothe–Rosenblatt transport from $\mu$ to $\nu$, while $\delta$ is a computational error that decreases as the running time of the transport algorithm increases. 
For this result, we need $\nu$ to be ``sequentially samplable'' with a ``bounded average sampling cost'' which is a novel but natural notion of independent interest. In addition, we prove the following.

\begin{itemize}
\item We prove an algorithmic version of the celebrated Talagrand's  inequality for transporting the standard Gaussian distribution $\Phi^n$ to an arbitrary $\nu$ under the Euclidean-squared cost.
 When $\nu$ is $\Phi^n$ conditioned on a set $\cS$ of measure $\eps$, we show how to implement the needed sequential sampler for $\nu$ in expected time $\poly(n/\eps)$, using membership oracle access to $\cS$. Hence, we obtain an algorithmic transport that maps $ \Phi^n$ to $ \Phi^n|\cS$ in time $\poly(n/\eps)$ and expected Euclidean-squared distance $O(\log 1/\eps)$, which is optimal for a general set $\cS$ of measure $\eps$.
 
\item As corollary, we find the first \emph{computational} concentration (Etesami et al. SODA 2020) result for the Gaussian measure under the Euclidean distance with a \emph{dimension-independent} transportation cost, resolving a question of Etesami et al. More precisely, for any set $\cS$ of Gaussian measure $\eps$, we map most of $\Phi^n$ samples to $\cS$ with Euclidean distance $O\big(\sqrt{\log 1/\eps}\big)$ in time $\poly(n/\eps)$.
\end{itemize}

%% file: intro.tex
\section{Introduction} 

Optimal transport (OT) is a fundamental problem that arises in mathematics, science, and engineering, including differential geometry~\cite{Figalli2011}, transportation planning~\cite{Santambrogio2009}, economics~\cite{Galichon2021}, machine learning~\cite{Montesuma2023, peyre2019computational},   image registration~\cite{haker2004optimal}, and seismic tomography~\cite{Metivier2016}. 
In machine learning, it has been used in unsupervised learning~\cite{YangTabak2021}, as a measure of the cost of misclassification~\cite{Frogner2015}, 
to define the fairness of algorithms~\cite{buyl2022optimal},
in Wasserstein GANs 
\cite{Arjovsky2017WassersteinGA}, for transfer learning~\cite{courty2017joint}, and in diffusion generative models~\cite{you2024renormalization, kim2024improving}. 
%

In the optimal transport problem, we would like to transport samples from a source distribution $\mu$ to points in the target distribution $\nu$ with a minimum expected ``transportation cost''  $\cost(x,y)$   of transporting $x\sim \mu$ to $y \sim \nu$. The study of this problem dates back to the work of Monge~\cite{monge1781memoire}, who wanted the transportation mapping $A(x)=y$ to be deterministic.
Kantorovich~\cite{kantorovich1942translocation} reformulated the problem by allowing $A(x)$ to be a randomized (stochastic) mapping.  In other words, we now look for a coupling $\pi$ over the distributions $\mu,\nu$ with minimum expected transportation cost $\Ex \cost(x,y)$, using which we define the optimal cost of transporting $\mu$ to $\nu$,
$$ \T(\mu,\nu) = \min_{\pi \in \cC} \Ex_{(x,y) \sim \pi} \cost(x,y) $$
where $\cC$ is the set of all couplings between $\mu,\nu$.
%
OT is   closely related to the notion of ``Wasserstein metric'' that generalizes  OT using a parameter $p\geq 1$ and is the same  for $p=1$. 

 As a prominent example of the use of OT in mathematics, Talagrand~\cite{talagrand1996transportation} gave a bound on 
the optimal transport, under the $\ell_2^2$ cost, of the $n$-dimensional Gaussian measure $\Phi^n$ to an arbitrary distribution $\nu$ based on the KL-divergence of $\nu$ from $\Phi^n$.  Using this, he derived an essentially optimal concentration of measure result, showing that for any set $\cS$ of measure $\eps$ in $\Phi^n$, almost all of the measure $\Phi^n$ is within $\ell_2^2$ (minimum) distance $O(\ln \nf{1}{\eps})$ from $\cS$.

\parag{Computational OT.} 
Computational aspects of OT have recently become extremely important on their own~\cite{peyre2019computational}. In the most common formulation of ``computational OT'', we would like to compute or estimate $\T(\mu,\nu)$ efficiently. Computing $\T(\mu,\nu)$ is a key tool, e.g., for applications that use   OT   to quantify a loss   that allows one to know ``how far'' we are from a target goal
\cite{bhagoji2019lower,birrell2023adversarially,blanchet2022optimal}.
%
A common approach to computing $\T(\mu,\nu)$ is to work with empirical sample sets $\cS \sim \mu^m, \cT \sim \nu^m$, and find the best OT between the empirical distributions $U_\cS, U_\cT$ that are uniform over $\cS,\cT$ (e.g., see~\cite{hutter2021minimax,manole2024plugin} and the references therein).  This approximation converges to the quantity $\T(\mu,\nu)$ in the limit, and the OT between $U_\cS, U_\cT$ can be computed using the Hungarian algorithm  for minimum weighted matching 
\cite{kuhn1955hungarian}.
The popular iterative Sinkhorn algorithm solves a regularized version of the OT  problem
\cite{sinkhorn1967concerning}
but it also works with empirical sample sets,
that is,
i.i.d. samples from the distributions.
Using empirical samples, one does not rapidly converge to the optimal OT in some elementary cases. For example, to transport the uniform distribution on the $n$-dimensional unit cube to itself,
the OT between two $\poly(n)$-size empirical versions of the original distribution 
is   $\Theta(\sqrt{n})$ in $\ell^2_2$ distance
even though the actual OT cost is zero.

\parag{Statistical OT.} The above approach of using empirical samples  $\cS \sim \mu^m, \cT \sim \nu^m$ can in fact be used to approximate the \emph{transport map} itself from $\mu$ to $\nu$, as  in~\cite{hutter2021minimax,manole2024plugin}. For example, Brenier's theorem~\cite{brenier1991polar,knott1984optimal} asserts that under the $\ell_2^2$ cost and suitable conditions, a unique Monge mapping achieves optimal transport, and one can aim at approximating this deterministic mapping. This approach is sometimes known as \emph{statistical} optimal transport~\cite{chewi2024statistical}.  However, this approach needs exponential in $n$ samples for $n$-dimensional distributions to achieve good approximate results.  Some previous works like~\cite{hutter2021minimax,manole2024plugin} make improvements on this analysis by assuming further smoothness and structural  conditions on the distributions but the curse of dimensionality basically remains intact. More importantly, to the best of our knowledge, no previous work models the algorithmic aspect of searching for the transport map by limiting its algorithm to run in polynomial time over the size of the  input $x \sim \mu$. 

%

 \subsection{Our Contributions}
 
In a nutshell, our contributions are (1) formalizing a new theory of algorithmic transport,  (2) obtaining initial results on algorithmic transport for the  high-dimensional setting, and (3) obtaining applications for algorithmic transport, e.g., to algorithmic concentration of measure. Each of the items above has multiple aspects that are elaborated in the following.

\parag{Algorithmic Transport in Polynomial Time.} 
The common computational OT formulation aims to compute or approximate the optimal transportation cost $\T(\mu,\nu)$, yet it does not answer the key question of {how to algorithmically compute}   the \emph{transport map} efficiently over the size of the given input sample.
I.e., suppose that we are given a particular sample $x \sim \mu$ as input, and we would like to map it to $y \sim \nu$ as follows: (1) The mapping   shall be computed in polynomial time over the size of the input $|x|=n$.
(2) We would like to   control the expected cost of the transportation.
%
%
To highlight the subtle distinction between our formulation and traditional computational OT, in this work we use the term \emph{algorithmic transport} to refer to the task of computing a (randomized) mapping $A$ efficiently based on its input size $|x|$ (e.g., the dimension of $x$), such that $A(x) \sim \nu$, whenever $x \sim \mu$. 

Algorithmic transport, when done optimally, can be used to approximate OT cost efficiently as well. In particular, when the transportation cost is bounded by a constant, using $k=\Theta(\eps^{-2}$) independent samples  $(x_1,y_1),\dots,(x_k,y_k) \sim (x,A(x))^k$, the average $\Ex_{i} \cost(x_i,y_i)$ gives an $\eps$-approximation of the OT, with high probability.
However,  it is not clear how to do the reverse and obtain algorithmic transport from computational OT.

When $\mu,\nu$ are one dimensional, for natural (convex) costs such as $\cost(x,y)=|x-y|^p, p\geq 1$ one can find simple algorithms that simply use a ``monotone'' transportation plan~\cite{villani2021topics}.
\iffull (See   Lemma \ref{lem:AlgOpt1D}). \fi
Furthermore, when the distributions $\mu,\nu$ have small domains of size $k$, one can use algorithms based on min-cost flows to find a full description of the OT from $\mu$ to $\nu$ in $\poly(k)$ time~\cite{peyre2019course}. However, our focus is on the high-dimensional setting and finding $\poly(n)$-time computable mappings between distributions of dimension $n$ with \emph{super}-polynomial support.
%
We ask:
\begin{quote}
    \emph{If $\mu,\nu$ are $n$-dimensional distributions,  how can we find a  $\poly(n)$-time computable  transport map from $x\sim \mu$ to $y \sim \nu$ of a small/optimal  cost?}
\end{quote}
Formalizing and answering the question above in various contexts is the focus of our work. Our studies also  bear similarities to the  line of work on  approximating the total variation distance~\cite{bhattacharyya2023approximating,feng2023simple} as it coincides with OT under the Hamming distance.


\parag{Online Transport in High-Dimensional Setting.} In this work, we approach the main question above through a study of  
those high-dimensional transports in which the transporting algorithm $A$ produces $y = (y_1,\dots,y_n)$ from $x = (x_1,\dots,x_n)$ in an online manner. Namely, $A$ shall output $y_i$ based on $x_{[i]} = (x_1,\dots,x_i)$ and before receiving $x_{i+1}$. 
The so-called Knothe-Rosenblatt transport (KR transport for short)~\cite{knothe1957contributions,rosenblatt1952remarks} is an important online transport with two properties (1) its reverse is also online, and (2) it follows a ``greedy'' approach in each round by using a monotone mapping of dimension one. 
Our motivation for studying online transports is twofold: 
(1)~Despite being a \emph{restriction} on how the transport is done,   the ``online lens'' guides us towards algorithm development;
    (2)~In our eyes, information-theoretic study of online algorithms is interesting. In particular, in Section \ref{sec:prod}, we prove that the KR transport is optimal among all online transports when the source distribution is a product.



\parag{Main Result: Algorithmic Transport from Product Distributions.}  
Our main  result (Theorem \ref{thm:main}) is to design a $\poly(n)$-time online  algorithm that transports a generic product distribution $\mu = \mu_1 \otimes \dots \otimes \mu_n$ to any $n$-dimensional distribution $\nu$, assuming that (1) the transportation cost satisfies $\cost(x,y)= \sum_{i } \cost_i(x_i,y_i)$, where $x=(x_1,\dots,x_n),y=(y_1,\dots,y_n)$, and (2) the transporting algorithm $A$ has oracle access to proper samplers for both  $\mu,\nu$. 

{The algorithm itself is simple:} {Given $x$, having determined $y_1$,\ldots, $y_{i-1}$, to determine $y_i$, it samples $k-1$ samples besides $x_i$ according to $\mu_i$. Similarly it samples $k$ samples according to the conditional distribution of the $i$th coordinate of $\nu$ conditioned on the values of $y_1, \ldots, y_{i-1}$. Then it optimally matches the two sets of two $k$ samples. The value of $y_i$ is the match of $x_i$ in this matching. }

The transportation cost of $A$ turns out to be $\Delta+\delta$, where $\Delta$ is the optimal cost of online transports from $\mu$ to $\nu$ (which, as we will prove, will coincide with the KR transport ~\cite{knothe1957contributions,rosenblatt1952remarks} in our settings of interest), and $\delta$ is a term that could be made  smaller by choosing $k$ larger. We show that the reverse transport from $\nu$ back to the product  $\mu$ can also be computed algorithmically.
This will be  useful for deriving further algorithmic transports through composition.

\parag{Sequential Samplers.} When it comes to the samplability conditions needed in our main results above, we merely require that we can sample from $\mu_i$ efficiently. However, for the non-product distribution $\nu$, the samplability condition is stronger and we require that one can sample from $\nu_i$ conditioned  on any previously sampled prefix $y_{[i-1]}$. We refer to such samplers as \emph{sequential} samplers. A key quantity of interest is the complexity of iterative sampling of the coordinates $y_1,\dots,y_n$ sequentially (conditioned on previous ones) till we obtain a full sample $y$. We would like to have samplers where the average complexity of this sequential generation is bounded.
As it turns out, we can indeed bound such costs in our special cases of interest.

From a real-world application point of view,  this notion of efficient sequential sampler is very natural in some generative models. This is indeed the case for transformer-based  language models that autoregressively generate their tokens one by one, each conditioned on the previously sampled sequence of tokens~\cite{vaswani2017attention,ford2018importance}. That is, the joint distributions produced by these generative models have sequential samplers of low expected cost, as they indeed generate their sequence of symbols in a reasonable time and in an online fashion.



%


\parag{Algorithmic Transport for the Standard Gaussian Distribution.} One of the fundamental results in OT is Talagrand's transportation inequality for the $n$-dimensional Gaussian distribution $\Phi^n$~\cite{talagrand1996transportation}. It is proved that for every distribution $\nu$, $\T(\Phi^n,\nu) \leq 2 \KL{\nu}{\Phi^n}$, in which the cost is measured in $\ell_2^2$, i.e., $\cost(x,y) = \sum_{i \in [n]} |x_i-y_i|^2$, and $\KL{\cdot}{\cdot}$ denotes the Kullback–Leibler divergence. In this work, we lift this classical result to the algorithmic setting.  Note that, as mentioned in~\cite{talagrand1996transportation}, this bound is optimal \emph{in general}, e.g., when $\nu$ is a shifted $\Phi^n$, in which case our results converge to this optimal bound as well. In particular, we derive this result from our main result by proving the following two complementary claims:
\begin{itemize}
    \item {\em Information theoretic:} We observe that Talagrand's  bound of $2\KL{\nu}{\Phi^n}$ upper bounds not only the best ``offline'' transport from the standard Gaussian $\Phi^n$, but also the best \emph{optimal online} transportation of $\Phi^n$ to $\nu$. Namely, we show that $\Delta \leq 2\KL{\nu}{\Phi^n}$, where $\Delta$ is the optimal online transportation cost as defined above.
    \item {\em Computational:} We use results from~\cite{fournier2015rate} to show that the Gaussian distribution in one dimension has a small transportation cost to its empirical samples on average.
\end{itemize}

\remove{
\parag{Online Algorithmic Transport from Gaussians to Gaussians.} It is well-known that optimal transport between high-dimensional Gaussians has a closed form formula~\cite{peyre2019computational}. We then study this question for the \emph{online} setting. We first observe that our result above does indeed achieve the optimal online transport from the standard Gaussian to any other Gaussian. Next, because we could obtain transports from and to the standard Gaussian, one natural idea   is to online-transport an arbitrary Gaussian $\mu$ to $\Phi^n$ and then compose this with an online transport from $\Phi^n$ to $\nu$. We analyze this composition directly and compute its cost, which improves upon the cost obtained through a union bound and Talagrand's upper bound based on KL-divergence for each of the two mappings. See Theorem \ref{thm:Gaussian-Online} for details.
}

\parag{Transporting Standard Gaussian to Conditional Gaussian.}   We show that in a natural setting, where $\nu $ is the Gaussian distribution conditioned on an event $\cS$ of Gaussian measure $\eps$, such sequential samplers can be efficiently simulated using oracle access to membership tests in $\cS$.  In other words, we find an algorithmic oracle-aided transportation algorithm that \emph{simultaneously} work for all such distributions $\nu = \Phi^n | \cS$. Note that such distributions have $2 \KL{\nu}{\Phi^n} \leq 2 \ln \nf{1}{\eps}$. We obtain  algorithmic transports running in expected time  $\poly(n/\eps)$  that achieve transport cost that converges to the upper bound of Talagrand. 

\newcommand{\dis}{\mathrm{d}}
\parag{Dimension-Independent Computational Concentration
for Gaussian Spaces.} 
One of the applications of OT is to obtain concentration of measure (CoM) inequalities~\cite{gozlan2010transport}: One shows that any set $\cS$ of ``sufficiently large'' measure in a concentrated metric probability space $(\mu,\dis)$, where $\mu$ is a distribution and $\dis$ is a distance metric, expands to cover most of the measure in $\mu$, when we consider neighbors of $\cS$ within a certain distance. Recently, a \emph{computational} (algorithmic) variant of the CoM phenomenon has been introduced~\cite{mahloujifar2019can,etesami2020computational}, in which one aims to show that the reverse mapping can be computed efficiently from almost all of the points in $\mu$ back to $\cS$ by moving the points within a bounded distance. Namely, given a typical sampled point $x \sim \mu$, we aim to algorithmically find a ``close neighbor'' $y \in \cS$ of bounded distance $\dis(x,y)$. The work of~\cite{etesami2020computational} obtained such results  for various settings, but their work left open obtaining computational CoM with dimension-independent (optimal) distance  for the basic and natural space of Gaussian distributions under the $\ell_2$ distance. Using our oracle set-transportation result for Gaussian spaces mentioned above, we  resolve this open question and obtain such an optimal and dimension-free bound (see Corollary \ref{cor:CoM}).

\parag{Reductions for (Deriving New) Algorithmic Transport.}
Finally, considering the role of reductions in resolving algorithmic tasks, we also develop the (right) notion of algorithmic reductions for the goal of relating algorithms for (optimal) transport across different spaces. In particular, suppose   $\mu_1,\mu_2$ are distributions and $\cost_1,\cost_2$ are two different transportation costs.  In 
\iffull
Definition \ref{def:reduction} 
\fi
we state conditions under which, we can   automatically  transform an algorithmic transport result from $\mu_1$ to $\nu$ (under the cost $\cost_1)$ to a similar result that transports $\mu_2$ to $\nu$ (under the cost $\cost_2$) for specific distributions $\mu_1,\mu_2$ and arbitrary distribution $\nu$. We then show how to realize such reductions when we transport uniform distributions over the unit cube  and the unit sphere (to an arbitrary distribution) by reducing them to  the case of transporting Gaussian distributions. Consequently, we obtain algorithmic transports from these distributions as well.
\iffull
(See Corollary \ref{cor:unit} and Theorem \ref{thm:sphere}).
\fi

%% file: Definitions.tex
\section{Basic Concepts} \label{sec:definitions}

In this section, we define the key notions studied in this paper. In
Section \ref{sec:properties} we make some useful remarks, but we have moved them to that section to keep the reading easier.

\parag{Notation.} We let $[n]=\set{1,\dots,n}$. We denote the source (initial) distribution as $\mu$.  When $\mu$ is distributed over $\R^n$, we say that $\mu$ has dimension $n$ and by  $\mu_i$ we denote the  distribution of its $i$th coordinate. We usually denote $x \sim \mu$, where $x=(x_1,\dots,x_n)$ and $x_i \sim \mu_i$.  $ \src = \mu_1 \otimes \dots \otimes \mu_n$ means that $\src$ is a product distribution.     We use a similar notation for the target distribution $\trg$.  By $y \gets A(x)$ we denote the process of running a probabilistic algorithm $A$ on input $x$  to obtain output $y$. When $\mu$ is a distribution, $A^\mu$ denotes an oracle algorithm  $A$ that has access to fresh samples from $\mu$, and when $\cS$ is a set, $A^\cS$ denotes a similar situation where the oracle responds membership in $\cS$. 
For vector $(v_1,\dots,v_n)$, by $v_{[i]}$ we denote the prefix vector $(v_1,\dots,v_i)$. When a distribution $ \mu$ of dimension $n$ with marginals $\mu_1,\dots,\mu_n$ is clear from the context, by $\mu_i|x_{[i-1]}$, we denote the distribution of $\mu_i$ conditioned on having sampled $x_j$ from $\mu_j$ for all $j<i$. For further clarity on the underlying joint distribution, we might sometimes use $\mu_i|_{\mu}x_{[i-1]}$ instead.
By $\mu(\cS)$ or $\Pr_\mu[\cS]$ we denote the probability of event $\cS$ under the distribution $\mu$. Whenever it is clear from the context, for an outcome $x$, we use $\mu(x)$ to either denote the probability of the outcome $x$ or the density of $\mu$ at $x$ depending on whether $\mu$ is discrete or continuous. By $\Supp(\mu)$ we denote the support set of $\mu$, which for the discrete and continuous cases can be defined as $\set{x \mid \mu(x)>0}$.
When $\Supp(\mu)\cup \Supp(\nu) \se \cS$, their Kullback–Leibler (KL) divergence is denoted as $\KL{\nu}{\mu}=\sum_{x \in \cS} \nu(x) \log \left({\nu(x)}/{\mu(x)}\right)$ with the  natural logarithm basis. In the preceding definition and generally throughout this paper, we use the summation notation corresponding to discrete distributions; the corresponding results for continuous distributions replace sums with proper integrals. For $p\geq 1$, the $\ell_p$-norm and $\ell_p$-distance over $\R^n$ are defined as $ \ell_p(x) =\norm{x}_p  = \big( \sum_{i\in[n]} |x_i|^p \big)^{1/p}$,    and  $\ell_p(x,y)=\ell_p(x-y)$.


\parag{Transportation Costs.} In the following, all transportation \emph{costs}, usually denoted as $\cost$, are functions $ \cost \colon \R^{2n} \To \R$ with non-negative   outputs that model the cost of transporting $x \sim \mu$ to $y \sim \nu$. We always assume $\cost$ to be lower semi-continuous but do not assume $\cost$ to be symmetric or satisfy the triangle inequality; we state these conditions, whenever needed.

\begin{definition}[Coupling and Optimal Transportation Cost] \label{def:transport}
We say that a distribution $\pi$ over pairs with marginals $\pi_1,\pi_2$ is a coupling of $\mu,\nu$ if $\pi_1=\mu,\pi_2=\nu$. If  for every $x \sim \mu$, there is a unique $y$ such that $(x,y) \in \Supp(\pi)$, then we call this a deterministic (Monge)   transport from $\mu$ to $\nu$.  
    For a  cost $\cost$, the transport cost of a coupling  $\pi$ of $\mu,\nu$   is defined as 
    $$\T_{\cost}(\pi) = \Ex_{(x,y)\sim \pi}  \cost(x,y).$$
    We refer to $\T^{1/p}_{\cost^p}(\pi) $ as the (Wasserstein) $p$-cost of $\pi$ under $\cost$. 
%
    If $\cC(\mu,\nu)$ denotes the set of all couplings between $\mu,\nu$,  the  (Kantorovich) \emph{optimal transportation cost} for $(\mu,\nu)$ is defined as 
$$\T_{\cost}(\mu,\nu) = \inf_{\pi \in \cC(\mu,\nu)} \T_{\cost}(\pi).$$
\end{definition}

The infimum in Definition \ref{def:transport} for defining the optimal transportation costs turns out to be a minimum as $\cost$ is lower-semi continuous~\cite{ambrosio2008gradient}.

\begin{definition}[Algorithmic Transport] \label{def:algT} For distributions $\mu,\nu$, algorithm $A$ is a \emph{transport} from distribution $\mu$ to distribution $\nu$ if $A$ is a (probabilistic) algorithm such that $A(x) \sim \nu$ whenever $x \sim \mu$. By $\pi_A$ we denote the coupling created by $A$. For a transportation cost $\cost$   the   transportation cost of $A$ is defined as  $\T_{\cost}(A) =\T_{\cost}(\pi_A)$.

\end{definition}


\parag{Computational Model.} In Definition \ref{def:algT}, we need to either work with discrete distributions with samples of finite length, or when the distributions are continuous we need to work with the generalization of \emph{algorithms working with real numbers} as formalized in~\cite{blum1998complexity,braverman2005complexity}. 
%

 

We now define online transport and its algorithmic variant.

\begin{definition}[Online (Algorithmic) Transport] \label{def:onlineT} For   distributions $\src,\trg$  of dimension $n$, we call a (probabilistic and perhaps computationally unbounded) algorithm $A$ an \emph{online transport} algorithm from   $\mu$ to   $\nu$ if it forms a transport from $\src$ to $\trg$, while it makes its decisions in an online way. Namely, $A$ has an internal iterating process (for simplicity also denoted by $
A$) that reads $(x_1,\dots,x_n) \sim \src$ coordinate by coordinate while holding an internal state, initially $s_0=\es$. In the $i$th iteration, we have
  $(s_i,y_i)\gets A(s_{i-1},x_i)$, and at the end we output $(y_1,\dots,y_n)\sim \nu$.
  We also let $\cC^\OnT(\mu,\nu)$ to be the set of all couplings that can be obtained by  online algorithms and for a transport cost $\cost$ obtain the optimal online transportation cost as
  $$\T^{\OnT}_{\cost} (\mu,\nu) = \inf_{\pi \in \cC^\OnT(\mu,\nu)}\T_{\cost} (\pi).$$  
\end{definition}

To contrast and emphasize on a transport not being necessarily online, we refer to (potentially) non-online transports as \emph{offline} transports.

We now define a class of   couplings that is closely related to online transport.
 
\begin{definition}[Online Coupling] \label{def:onlineCoup}
Suppose $\pi$ is a coupling between $n$-dimensional distributions $\mu,\nu$, and $\pi_i$ is the corresponding marginal coupling between $\mu_i,\nu_i$. We call $\pi$ an   \emph{online coupling} if for all $z=(x_{[i-1]},y_{[i-1]}) \in \Supp(\pi_{[i-1]})$,   
$\pi_i | z$ is a coupling of   $\mu_i|  x_{[i-1]}$ (according to $\src$) and $\nu_i|y_{[i-1]}$ (according to $\trg$). 
If $\cC^\OnC(\src, \trg)$ denotes the set of all online couplings between $\src,\trg$, for a transport cost $\cost$ we obtain   the optimal online coupling cost between $\mu,\nu$ as 
$$\T^\OnC_{\cost}(\src,\trg) = \inf_{\pi \in \cC^\OnC(\mu,\nu)}\T_{\cost}(\pi).$$
\end{definition}

We now show how to characterize online couplings using online transports.

\begin{proposition} \label{prop:onlineVSbidir}
    A coupling ${\pi}$ between $\src,\trg$ is \emph{online} if and only if it can   be obtained through both an online transport from $\src$ to $\trg$ and an online transport  from $\trg$ to $\src$.
\end{proposition}
\iffull
\begin{proof}[Proof of Proposition \ref{prop:onlineVSbidir}]
    Suppose ${\pi}$ is an online coupling. Then we show how it can be obtained through an online transport from $\src$ to $\trg$ (and a similar argument works from $\trg$ to $\src$). The online algorithm receives $x_i \sim (\mu_i | x_{[i-1]})$ and it samples from $(\pi_i|x_{[i]},y_{[i-1]})$. By definition of $\pi$ being online, sampling $x_i \sim (\mu_i |  x_{[i-1]})$ is the same as sampling 
    $x_i \sim (\mu_i |  x_{[i-1]},y_{[i-1]})$, and hence $(x_{[i]},y_{[i]})\sim \pi_{[i]}$.

    Now, suppose ${\pi}$ is online in both directions. Then, conditioned on the result of the first $i$ rounds being $x_{[i-1]},y_{[i-1]}$, we claim that $(\pi_i|x_{[i-1]},y_{[i-1]})$, when projected into its two components, leads to the following two marginal distributions: $\mu_i|x_{[i-1]}$ and $\nu_i | y_{[i-1]}$. Note that if we prove this claim, it means that $(\pi_i|x_{[i-1]},y_{[i-1]})$ is a coupling between these two distributions, and hence it is online. The reason for this claim is as follows. First, one marginal should be $\mu_i|x_{[i-1]}$ because ${\pi}$ is online from $\src$ to $\trg$, and therefore it works by receiving $x_i \sim \mu_i|x_{[i-1]}$ and outputting a sampled $y_i$. A similar argument holds for the second marginal being $\nu_i | y_{[i-1]}$ because ${\pi}$ is also online from $\trg$ to $\src$.
\end{proof}
\fi
\begin{definition}
   We call the    cost function $\cost$ over $\R^n \times \R^n$ \emph{linear} over $\cost_1,\dots,\cost_n$, if
   $  \cost( {x},  {y}) =  \cost_1(x_1,y_1)+ \dots + \cost_n(x_n,y_n),$
   for all $x=(x_1,\dots,x_n), y=(y_1,\dots,y_n)$.
\end{definition}

\parag{Greedy Coupling.} One might wonder how we can compute/approximate $\T^\OnC_{\cost}(\src,\trg)$. One approach is to use greedy methods, by trying to use an optimal coupling in each round. This is formalized in the following definition in settings with dedicated costs for each round. We will then discuss when this method succeeds in Theorem \ref{thm:prod}.
More generally, we  define  locally-optimal couplings, even when they are not online.
\begin{definition}[Locally Optimal and Greedy   Couplings] \label{def:greedy}
Suppose the cost function $\cost$ over $\R^{2n}$ is linear over $\cost_1,\dots,\cost_n$. A  coupling ${\pi}$    between $\src,\trg$ is     \emph{locally optimal}, if for every $z_{[i-1]} \in \Supp(\pi_{[i-1]})$,  it holds that $ \pi_i | z_{[i-1]}$ is an optimal transport between  $\mu_i|  z_{[i-1]},\nu_i|z_{[i-1]}$ according to $\cost_i$;   i.e.,
$$\T_{\cost_i}(\pi_i|z_{[i-1]}) =\T_{\cost_i}(\mu_i|  z_{[i-1]},\nu_i|z_{[i-1]}).$$
When $\pi$ is an online coupling as well,   the above condition simplifies to 
$$\T_{\cost_i}(\pi_i|z_{[i-1]}) =\T_{\cost_i}(\mu_i|  x_{[i-1]},\nu_i|y_{[i-1]}),$$ 
in which case we call $\pi$ \emph{greedy}. For $\cC^\OnG(\mu,\nu)$ denoting the set of all greedy couplings from $\mu$ to $\nu$, we define
$$ \T^{\OnG}_\cost(\mu,\nu) = \sup_{\pi \in \cC^\OnG(\mu,\nu)} \T_{\cost}(\pi).$$

\end{definition}
%


\begin{remark}[Greedy  vs. Knothe-Rosenblatt Transports] \label{rem:Greedy-KR}
Greedy couplings are   closely related to Knothe-Rosenblatt (KR for short) transports \cite{knothe1957contributions,rosenblatt1952remarks}. Specifically, for a greedy coupling $\pi$, when the cost functions $\cost_i$ are convex,  for any $z_{[i-1]}\sim \pi_{[i-1]}$, the locally optimal coupling $\pi_i|z_{[i-1]}$ could be obtained by simply using the unique monotone mapping \cite{carlier2010knothe}. 
\iffull
(See  Lemma \ref{lem:AlgOpt1D}.) \fi
Hence, KR coupling is a special case of greedy couplings and cover many interesting cases in this class. For example, when the cost function $\cost$ is $\ell_p^p$ for $p\geq 1$, then $\T^{\OnG}_\cost(\mu,\nu)$ equals the cost of the KR coupling between $\mu$ and $\nu$. However, due to the generality of  greedy couplings (e.g., for non-monotone costs) we define and use greedy transports.
\end{remark}

\parag{Lambda and Delta Cost Functions.} We now define two functions that play key roles in our analysis of the cost of online transports. The first (Lambda) function depends on a coupling, while the second one (Delta) depends on the two distributions that are coupled. As we prove later in Proposition \ref{prop:props}, Lambda is a parameter that  lower bounds the cost of any coupling. Delta  is the optimal online transport from a product distribution to another one.

\begin{definition}[The Lambda and Delta  Functions] \label{def:lower-bound}
    For a coupling  $\pi$ of dimension $n$ between distributions $\mu,\nu$ of dimension $n$,  and a cost function $\cost$ that is   linear over $\cost_1,\dots,\cost_n$, we define the 
    \emph{Lambda  functions} as
$$\lowerbound_{\cost}(\pi) = \Ex_{z\sim \pi}\sum_{i\in[n]}\T_{\cost_i}(\mu_i|z_{[i-1]}, \nu_i|z_{[i-1]}).$$
We also define the \emph{Delta  function} between distributions $\mu,\nu$ of dimension $n$  as 
$$\lowerprod_{\cost}(\mu,\nu) = \Ex_{y\sim \nu}\sum_{i\in[n]}\T_{\cost_i}(\mu_i,\nu_i|y_{[i-1]}).$$
\end{definition}

Note that the coupling $\pi$ in Definition \ref{def:lower-bound} does not have to be online. Furthermore, the definition of $\lowerbound(\cdot)$ does depend on the order of the coordinates of the $n$-dimension  distributions.

\subsection{Characterizing Online Transport from Products} \label{sec:prod}
We end this section by stating a theorem showing that, whenever $\src$ is product,  any online coupling that is ``locally optimal'' in the sense that given history $z=(x_{[i-1]},y_{[i-1]})$ it finds (an arbitrary) optimal transport between $(\mu_i),(\nu_i|y_{[i-1]})$, finds an optimum online coupling between $\src,\trg$  as well as an optimal online transport from $\src$ to $\trg$. This theorem does not assume convexity of the costs. 
 As stated in Remark \ref{rem:Greedy-KR}, for convex transportation  costs, greedy algorithms can be instantiated using the KR transform. 

\begin{theorem}[Optimal Online Coupling and Transport    from Products] \label{thm:prod}
If $\src=\mu_1\otimes \dots \otimes \mu_n$ is product and the    cost  function $\cost$   is linear over $\cost_1,\dots,\cost_n$, then 
$$ \T^\OnT_{\cost}(\src,\trg) = \T^\OnC_{\cost}(\src,\trg)= \T^\OnG_{\cost}(\src,\trg)= \lowerprod_{\cost}(\src,\trg).$$
\end{theorem}

Before proving Theorem \ref{thm:prod} we prove some basic tools that are used in the proof. The first lemma that we state can be obtained from a simple application of the linearity of expectation.

\begin{lemma}[Cost Splitting] \label{lem:switching}
    Let $\pi$ be a coupling between distributions $\mu,\nu$ of dimensions $n$, and let  $\pi_i$ be the corresponding coupling between the marginals $\mu_i,\nu_i$. Suppose $\cost$ is linear   over $\cost_1,\dots,\cost_n$, and $\omega$ is an $n$-dimensional distribution that is arbitrarily correlated with $\pi$. Then,
    $$\T_{\cost}(\pi) = \sum_{i \in [n]} \T_{\cost_i}(\pi_i) 
    = \Ex_{z \sim \omega  } \sum_{i \in [n]}  \T_{\cost_i}(\pi_i|\omega_{[i-1]}=z_{[i-1]}).$$
    In particular, we can choose $\omega=\nu$, $\omega=\mu$, or   $\omega=\pi$ as special cases. 
\end{lemma}

We now prove basic properties of the two functions, showing how they can be used and characterized in special settings. 
In summary, Lambda function  lower bounds the transportation of every coupling, while Delta will play a key role in characterizing the transportation cost   for product distributions.

\begin{proposition}[Properties of Lambda and Delta Functions] \label{prop:props}
Suppose $\pi$ couples $\mu,\nu$ and $\cost$ is linear. The Lambda   function satisfies the following properties.
\begin{enumerate}
    \item \label{prop:1} {\em Lower Bound:} For all $\pi$,   $\lowerbound_{\cost}(\pi) \leq \T_{\cost}(\pi)$, and the equality holds iff $\pi$ is locally optimal.
    \item \label{prop:2} {\em Online Transports from Products:} If $\pi$ is an online transport and $\mu = \mu_1\otimes \dots \otimes \mu_n$, then
    $$\lowerbound_{\cost}(\pi) \geq  \lowerprod_{\cost}(\mu,\nu) .$$
   \item \label{prop:3} {\em Online Coupling for Products:} If $\pi$ is   an online \emph{coupling}, and $\mu$ is product then 
    $$\lowerbound_{\cost}(\pi) = \lowerprod_{\cost}(\mu,\nu) .$$
    \end{enumerate}
\end{proposition}

\begin{proof} [Proof of Proposition \ref{prop:props}]
We prove the claims in order.
\begin{enumerate}
    \item By   letting $\omega=\pi$ in Lemma \ref{lem:switching}, we get
    $$\T_{\cost}(\pi) =  \Ex_{z\sim \pi} \sum_{i \in [n]} \T_{\cost_i}(\pi_i|z_{[i-1]}) \geq \Ex_{z\sim \pi} \sum_{i \in [n]} \T_{\cost_i}(\mu_i|z_{[i-1]},\nu_i|z_{[i-1]}) = \lowerbound_{\cost}(\pi),$$
    where the  inequality follows from the fact that $\T_{\cost_i}(\cdot,\cdot)$ minimizes the transportation cost.
    \item We first claim that, in this case, for every $z_{[i-1]} = (x_{[i-1]},y_{[i-1]}) \sim \pi _{[i-1]}$, we have $\mu_i|z_{[i-1]} = \mu_i$. This is true, because (1) $(\mu_i | x_{[i-1]},y_{[i-1]}) = (\mu_i|x_{[i-1]})$ and the fact that $\pi$ is an online transport, and (2) $(\mu_i|x_{[i-1]}) = \mu_i$ by the fact that $\mu$ is a product. Therefore,  
    $$\lowerbound_{\cost}(\pi) =  \Ex_{z\sim \pi}\sum_{i\in[n]}\T_{\cost_i}(\mu_i|z_{[i-1]},\nu_i|z_{[i-1]})
    =\Ex_{z=(x,y)\sim \pi}\sum_{i\in[n]}\T_{\cost_i}(\mu_i,\nu_i|z_{[i-1]})
    .$$
    We now use the right hand side. In analyzing the right hand side, we first use Lemma \ref{lem:switching} (using $\omega = \pi$) and then sample $x,y$ in reverse   order,
    $$\Ex_{(x,y)\sim \pi}\sum_{i\in[n]}\T_{\cost_i}(\mu_i,\nu_i|z_{[i-1]}) = 
    \sum_{i\in[n]} 
    \Ex_{y_{[i-1]} \sim \nu_{[i-1]}} 
    \Ex_{ x_{[i-1]} \sim \nu_{[i-1]}|y_{[i-1]}}    \T_{\cost_i}(\mu_i,\nu_i|y_{[i-1]},x_{[i-1]}),
    $$
    where for each $i \in [n]$, we sample $(x_{[i-1]},y_{[i-1]} )\sim \pi_{[i-1]}$ by first sampling $y_{[i-1]}$ and then sampling $x_{[i-1]}$ conditioned on $y_{[i-1]}$. Now, for every $y_{[i-1]} \sim \nu_{[i-1]}$, we claim that 
 $$ \Ex_{ x_{[i-1]} \sim \nu_{[i-1]}|y_{[i-1]}}    \T_{\cost_i}(\mu_i,\nu_i|y_{[i-1]},x_{[i-1]}) \geq \T_{\cost_i}(\mu_i,\nu_i|y_{[i-1]}).$$ 
 This claim follows from Part \ref{avearge:2} of Proposition \ref{prop:average} and the fact that the average of $\nu_i|y_{[i-1]},x_{[i-1]}$ over the choice of $x_{[i-1]} \sim \nu_{[i-1]}|y_{[i-1]}$ is $\nu_i|y_{[i-1]}$.

    \item When the coupling  $\pi$ is further an online coupling, then the equality holds, because 
    $$(\nu_i|y_{[i-1]},x_{[i-1]}) = (\nu_i|y_{[i-1]}),$$ and the   inequality above becomes an equality.
\end{enumerate}
\end{proof}

\begin{proof}[Proof of Theorem \ref{thm:prod}]
It is enough to prove the following two claims.
\begin{enumerate}
    \item $\T^\OnG_\cost(\src,\trg) \leq \lowerprod_{\cost}(\src,\trg)$.
    \item $\T^\OnT_\cost(\src,\trg) \geq \lowerprod_{\cost}(\src,\trg)$.
\end{enumerate}
The reason is that we already know $\T^\OnT_\cost(\src,\trg) \leq \T^\OnG_\cost(\src,\trg)$ (as being greedy is a limitation), and so proving the two claims above would imply all the equalities of the theorem statement. 

To prove the first claim, we observe that cost $\lowerprod_{\cost}(\src,\trg)$ can be achieved using (any) greedy algorithm   that (by definition) optimally couples $\mu_i = \mu_i|x_{[i-1]}$ with $\nu_i|y_{[i-1]}$ in the $i$th step. In fact,  \emph{all} greedy  coupling algorithms have the same cost $\lowerprod_{\cost}(\src,\trg)$ when one of the distributions  is product.

To prove the second claim, let $\pi$ be an online transport with cost $\T^\OnT_{\cost}(\mu,\nu)$. Our claim follows from Parts \ref{prop:1} and \ref{prop:2} of Proposition \ref{prop:props}, due to $\pi$ being online and $\mu$ being a product.
$$ \T^\OnT_{\cost}(\mu,\nu) =  \T_{\cost}(\pi) \geq \lowerbound_{\cost}(\pi) \geq \lowerprod_{\cost}(\mu,\nu).$$
\end{proof}

%% file: BasicTools.tex
\section{Basic Tools} \label{sec:tools}

\subsection{Composition and   Triangle Inequalities} 

\parag{Multi-distribution Coupling and Composition.} We now generalize the notion of coupling to more than two distributions and use it to define composition of (online) couplings.

\begin{definition}[Multi-distribution Coupling] A coupling 
$\pi$ of $\mu_1,\dots,\mu_n$ is a distribution over $n$-vectors such that the marginal of the $i$th coordinate is distributed as $\mu_i$. 
\end{definition}

\begin{definition}[Composition of Couplings] \label{def:composeCoupling}
 For coupling $\pi_{1,2}$ over $\mu_1,\mu_2$ and coupling $\pi_{2,3}$ over $\mu_2,\mu_3$, we define  the composition    
 $\pi_{1,3}=\pi_{2,3} \circ \pi_{1,2}$ of $\pi_{1,2}$ and $\pi_{2,3}$ as  the marginal of the first and third coordinates of the  (unique) coupling of $\mu_1,\mu_2,\mu_3$ such that.
\begin{enumerate}
    \item For $1 \leq i<j \leq 3$, the marginal distribution of $(\mu_i,\mu_j)$ in $\pi_{1,2,3}$ is distributed as $\pi_{i,j}$.
    \item In the coupling $\pi_{1,2,3}$, $\mu_1,\mu_3$ are independent, conditioned on $x_2 \sim \mu_2$.
\end{enumerate}
\end{definition}

We now use Wasserstein  $p$-cost, to state the following well-known triangle inequality.
\begin{lemma}[Triangle Inequality for Wasserstein $p$-Costs] \label{lem:triangleW} 
Suppose a cost function $\cost$  satisfies the triangle inequality (but not necessarily symmetry) and $p \geq 1$.  Then, for every coupling $\pi$ over $\mu_1,\mu_2,\mu_3$ with marginal coupling $\pi_{i,j}, i<j$ over $\pi_i,\pi_{j}$, we have  the following,
$$\T^{1/p}_{\cost^p}(\pi_{1,3}) \leq  \T^{1/p}_{\cost^p}(\pi_{1,2}) + \T^{1/p}_{\cost^p}(\pi_{2,3}).$$
 \end{lemma}

The following proposition can be obtained from the triangle inequality of Lemma \ref{lem:triangleW}.

\begin{proposition}[Triangle Inequality for Wasserstein $p$-Costs in Multi-Round Settings] \label{prop:triangle}
    Let $\mu$ be a distribution over $\R^n$, and for every $i \in [n],x_{[i-1]} \in \Supp(\mu_{[i-1]})$ let $J(x_{[i-1]})$ be a distribution over \emph{triples} of distributions over $\R$. Suppose  $\cost$ satisfies the triangle inequality and $\cost^p$  is linear over $\cost_1,\dots,\cost_n$ for $p \geq 1$. Then, the following holds.
    $$\left(\Ex_{x \sim \mu} \sum_{i \in [n]} \Ex_{(\nu_1,\nu_2,\nu_3)\sim J(x_{[i-1]})} \T_{\cost_i}(\nu_1,\nu_3) \right)^{1/p} 
    \leq \sum_{k \in [2]} \left(\Ex_{x \sim \mu} \sum_{i \in [n]} \Ex_{(\nu_1,\nu_2,\nu_3 )\sim J(x_{[i-1]})} \T_{\cost_i}(\nu_k,\nu_{k+1} ) \right)^{1/p} 
    $$
\end{proposition}
\iffull
\begin{proof}
Consider the following joint distribution over $n$-vectors.
(1) Sample $x \sim \mu$, then for all $i \in [n]$. (2) Sample $(\nu_1,\nu_2,\nu_3) \sim J(x_{[i-1]})$. (3) Consider an optimal-transport coupling $\nu_{1,2}$ (resp. $\nu_{2,3}$) of cost $\T_{\cost_i}(\nu_1,\nu_2)$ ($\T_{\cost_i}(\nu_2,\nu_3)$). (4) Let $\nu_{1,3}$ be the composition of $\nu_{1,2}$ and $\nu_{2,3}$, and $\nu_{1,2,3}$ be the resulting coupling. We denote this as $\nu_{1,2,3} \gets J(x_{[i-1]})$ (Note that $\nu_{1,3}$ is not necessarily an optimal coupling of $\nu_1,\nu_3$.)  (5) Sample $(y^{1}_i,y^2_i,y^3_i) \sim \nu_{1,2,3}$. (6) Output a triple of vectors $(y^{1},y^{2},y^{3}) = (y^{1}_{[n]}, y^{2}_{[n]}, y^{3}_{[n]})$ distributed as $\pi^{1,2,3}$ with marginal couplings $\pi^{1,2},\pi^{1,3},\pi^{2,3}$, and where $\pi^{j,k}_i$ refers to the corresponding coupling only for the $i$th coordinate. 


For all $1\leq j < k\leq 3$, by Lemma \ref{lem:switching} applied to $\omega = \mu$ we have 
$$\T_{\cost^p}(\pi^{j,k}) = \Ex_{x  \sim \mu } \sum_{i \in [n]}  \T_{\cost_i} (\pi^{j,k}_i | x_{[i-1]})
= \Ex_{x \sim \mu} \sum_{i \in [n]} \Ex_{\nu_{1,2,3} \gets J(x_{[i-1]})} \T_{\cost_i}(\nu_{j,k}).$$
By definition of $\nu_{1,2,3}$, the right hand side of the inequality of the proposition is equal to $\T^{1/p}_{\cost^p}(\pi^{1,2}) + \T^{1/p}_{\cost^p}(\pi^{2,3})$, while the left hand side is \emph{upper bounded} by $\T_{\cost^p}(\pi^{1,3})$, because $\T_{\cost_i}(\nu_1,\nu3) \leq \T_{\cost_i}(\nu_{1,3})$. Hence, the result follows from the triangle inequality of Lemma~\ref{lem:triangleW} applied to $\pi^{1,2,3}$.
\end{proof}
\fi
The following  can be obtained from the definition of online transport and Lemma~\ref{lem:triangleW}.
\begin{lemma}[Properties of the Composition of Online Transports] \label{lem:onlineComposes}
Consider an online transport $A_{1,2}$ from $\mu_1$ to $\mu_2$ with coupling $\pi_{1,2}$ and  an online transport $A_{2,3}$ from $\mu_2$ to $\mu_3$ with coupling $\pi_{2,3}$. Let   $\pi_{1,3}=\pi_{2,3} \circ \pi_{1,2}$ be the composed coupling. Then,
\begin{enumerate}
    \item The coupling $\pi_{1,3}$ is an online coupling.  
    \item There is an algorithm $A_{1,3}$  that transports  $\mu_1$ to $\mu_3$ as the coupling $\pi_{1,3}$, whose complexity is bounded by  running $A_{1,2}$ followed by running $A_{2,3}$.
    \item If the cost function $\cost$ satisfies the triangle inequality, then for all $p \geq 1$ the following holds
$$ \T^{1/p}_{\cost^p}(A_{1,3}) \leq \T^{1/p}_{\cost^p} (A_{1,2}) + \T^{1/p}_{\cost^p} (A_{2,3}). $$
\end{enumerate}
\end{lemma}
\iffull
\begin{proof}
    The algorithm $A_{1,3}$ maintains state $u_{i}$ that is the concatenation of state $s_i$ for $A_{1,2}$ and $t_{i}$ for $A_{2,3}$. Given $x_i \sim (\mu_{1,i}|x_1,\dots,x_{i-1})$, $A_{1,3}$ first obtains $(s_{i},y_i) \gets A_{1,2}(s_{i-1},x_i)$  followed by 
$(t_i,z_i) \gets A_{2,3}(t_{i-1},y_i)$, and it finally outputs $z_i$. This proves the first two items.

The third item follows   from the triangle inequality for Wasserstein $p$-cost (Lemma \ref{lem:triangleW}).
\end{proof}
\fi
The first item in Lemma \ref{lem:onlineComposes} and Proposition \ref{prop:onlineVSbidir} together show that the composition of two online couplings is also an online coupling.

 \subsection{Transport Through Intermediate Distributions}
In this section, we describe a method of transporting $\mu$ to $\nu$ (perhaps in an online and iterative way) through optimal transports between intermediate distributions in one dimension. We start with some definitions. We start by defining the notion of average for distributions and stating a general way of transporting through averages.

\begin{definition}[Average Distribution] \label{def:average}
    Suppose $M$ is a distribution over distributions. We define \emph{the average} of $M$, denoted as $\Ex[M] = \Ex_{\mu' \sim M}[\mu'] = \mu$, to be the distribution $\mu$ of the random variable $x$ that is sampled by first sampling $\mu' \sim M$ and then $x \sim \mu'$. Namely,  $\mu$ is the distribution   that $\mu(\cS) = \Ex_{\mu' \sim M} \mu'(\cS)$ for all the events $\cS$ defined over $\cup_{\mu' \in \Supp(M)} \Supp(\mu')$.
\end{definition} 

\begin{proposition}[Transport to Averages] \label{prop:average}
    Suppose $M$ is a distribution over distributions with average $\mu$. 
    \begin{enumerate}
        \item \label{avearge:1} Suppose $\pi$ is the following joint distribution. We first sample $\mu' \sim M$, then couple $\mu'$ with $\nu$ as $\pi_{\mu'}$, and then output a sample $(x,y) \sim \pi_{\mu'}$. Then, $\pi$ is a coupling between $\mu,\nu$.
        \item \label{avearge:2} $\Ex_{\mu' \sim M} \T_{\cost}(\mu',\nu) \geq \T_{\cost}(\mu',\nu)$.
    \end{enumerate}
\end{proposition}
\begin{proof}
    Part \ref{avearge:1} holds because the marginals of $x$ and $y$ have the marginals of $\mu,\nu$. Part  \ref{avearge:2} follows from  Part \ref{avearge:1} and picking $\pi_{\mu'}$ to be the optimal transport between $\mu',\nu$.
\end{proof}

The following definition states a way of finding a transport from $\mu$ to $\nu$ by working with alternative (intermediate) distributions that have averages $\mu,\nu$. 

\begin{definition}[Transport Through Intermediate Distributions] \label{def:intermediate}
    Let $\mu,\nu$ be distributions, $\cost$ be a cost function, and 
    $J$ be a  distribution over pairs of distributions.  We say that  algorithm $A$ couples $\mu,\nu$ \emph{through (the intermediate distribution)} $J$, if the following conditions hold.  
    \begin{enumerate}
        \item    $J$ produces marginals with averages    $\mu,\nu$. I.e.,   $\mu = \Ex_{(\mu',\nu') \sim J} \mu'$ and $\nu = \Ex_{(\mu',\nu') \sim J} \nu'$.   
        \item Algorithm $A$ first samples  $(\mu',\nu') \sim J$, then  finds some \emph{optimal} transport $\pi$ between $\mu',\nu'$ according to   $\cost$, and finally outputs $(x,y) \sim \pi$.
    \end{enumerate}
\end{definition}

\begin{definition}[Conditioning and Composing Transports with Distributions]  \label{def:compose-dist-transp} Suppose $\mu',\mu,\nu$ are distributions and $\pi$ is a transport from $\mu$ to $\nu$. If $\Supp(\mu') \se \Supp(\mu)$, then consider  the following sampling process.
\begin{enumerate}
   \item Sample $x \sim \mu'$.
   \item Sample $y$ from the $\nu$-coordinate of $\pi$, conditioned on its $\mu$-coordinate being $x$.
\end{enumerate}
Then,  the notation $\pi|\mu'$ denotes the joint distribution of $(x,y)$ and  $\pi \push \mu'$ denotes  the distribution of $y$.
Additionally, if $M$ is a distribution over distributions, then   $N = \pi \push M$ denotes the distribution over distributions  sampled by outputting $\nu' = \pi \push \mu'$ for $\mu' \sim M$. 
\end{definition}

\parag{Notation.} Let $U_{k,\mu}$ be the distribution over distributions obtained by first sampling $\cX \sim \mu^k$, and then outputting $\mu' = U_\cX$. A simple observation is that $\Ex U_{k,\mu} = \mu$ for all $k$.

\begin{proposition} \label{prop:props-of-compose}
If $M$ is a distribution over distributions with average distribution $\mu$, and if $\pi$ is any transport from $\mu$ to $\nu$, then the following holds.
\begin{enumerate}
    \item \label{part:1} $N = \pi \push M$ is a distribution over distributions with average $\nu$. 
    \item \label{part:2} For cost $\cost$,
    $\T_\cost(\pi) = \Ex_{\mu' \sim M} \T_\cost(\pi|\mu')$ in which $\pi|\mu'$ is defined in Definition~\ref{def:compose-dist-transp}.
    \item \label{part:3}  $U_{k,\nu} = \pi \push U_{k,\mu}$, and if $\mu$ is samplable in time $t_\mu$ and  coupling $\pi$ is computable in time $t_\pi$, then one can sample the set $\cY, |\cY|=k$ that describes  $U_\cY \sim U_{k,\nu}$ in time $k (t_\mu + t_\pi)$.
\end{enumerate}
\end{proposition}

\begin{proof}
For Part \ref{part:1}, observe that if we sample $x \sim \mu'$ for $\mu' \sim M$,  by definition we get $x \sim \mu$, which means $y \sim \pi \push M$ will be sampled as $y \sim \nu$.

For Part \ref{part:2}, $\Ex_{\mu' \sim M} \T_\cost(\pi|\mu')$ also computes the cost of the same coupling $\pi$ by breaking it into marginal costs based on how $x \sim \mu$ is sampled.

For Part \ref{part:3}, let $(x,y)\sim \pi$.
We  first sample $(x_1,\dots,x_k) = \cX \sim \mu^k$ and then let $\cY = (y_1,\dots,y_k)$ for $y_i \sim y|x=x_i$. It holds that $x_i$s are independently sampled according to $\mu$, and because $\pi$ transports $\mu$ to $\nu$, $y_i$'s are also independently sampled according to $\nu$.
\end{proof}

\begin{lemma}[Multi-Round Algorithmic Coupling Through Intermediate Distributions] \label{lem:Multi-Round-Intermediate}
    Suppose  cost function  $\cost$ satisfies the triangle inequality, and $\cost^p$ is linear over $\cost_1,\dots,\cost_n$ for $p\geq 1$. Let $\pi$, with marginals $\pi_1,\dots,\pi_n$ be a transport from $\src$ with marginals $\mu_1,\dots,\mu_n$ to   $\trg$ with marginals $\nu_1,\dots,\nu_n$.
For round $i \in [n]$ and previously sampled $z_{[i-1]} = (x_{[i-1]},y_{[i-1]}) \in \Supp(\pi_{[i-1]})$, suppose
    $J(z_{[i-1]})$   is a distribution over pairs of distributions defined based on $z_{[i-1]}$, and $\sigma_{z_{[i-1]}}$  is an arbitrary transport from  $\mu_i|z_{[i-1]}$ to $\nu_i|z_{[i-1]}$ under $\cost_i$.
    Suppose $\pi$ can also be obtained using the following algorithm $A$  in $n$ rounds. In round $i \in [n]$ and for previously sampled $z_{[i-1]} = (x_{[i-1]},y_{[i-1]}) \in \Supp(\pi_{[i-1]})$, $A$   couples $\mu_i|z_{[i-1]}$ and $\nu_i|z_{[i-1]}$   {through} the intermediate distribution $J(z_{[i-1]})$ (as defined in Definition~\ref{def:intermediate}) using the cost $\cost_i$. 
    %
%
    Then, 
$$              \T^{1/p}_{\cost^p}(\pi) \leq  
\left(\Ex_{{z} \sim \pi} \sum_{i \in [n]} \Ex_{(\mu'_i,\nu'_i) \sim J(z_{[i-1]})} \T_{\cost_i}\left(  
            \mu'_i, \sigma^{-1}_{z_{[i-1]}} \push \nu'_i
            \right)  \right)^{1/p}
      + \left(\Ex_{{z} \sim \pi} \sum_{i \in [n]}  \T_{\cost_i}(\sigma_{z_{[i-1]}}) \right)^{1/p},
$$    
where $\sigma^{-1}$ refers to the inverse coupling that changes the order of its marginals.
\end{lemma}

\begin{proof}[Proof of Lemma \ref{lem:Multi-Round-Intermediate}]
    The proof uses the triangle inequality for Wasserstein $p$-costs for the multi-round setting (Proposition~\ref{prop:triangle}). 
    
    For each $i \in [n]$ and $z_{[i-1]} \in \Supp(\pi_{[i-1]})$, consider  the following sampling process $I(z_{[i-1]})$ that extends $J(z_{[i-1]})$ by outputting one more coordinate as well.
        \begin{enumerate}
            \item Sample  $(\mu',\nu') \sim  J(z_{[i-1]})$.
            \item Let $\mu'' = \sigma^{-1}_{z_{[i-1]}} \push \nu'$.
            \item Obtain $(\mu',\mu'',\nu') \sim I(z_{[i-1]})$.
        \end{enumerate}
It holds that $\T^{1/p}_{\cost^p}(A)   
        = \left(\Ex_{z \sim \pi} \sum_{i \in [n]} \Ex_{(\mu',\mu'',\nu')\sim I(z_{[i-1]})} \T_{\cost_i}(\mu',\nu')\right)^{1/p} $, which is the left side of the inequality of Proposition \ref{prop:triangle}, and the right side is:
$$ \left(\Ex_{z \sim \pi} \sum_{i \in [n]} \Ex_{(\mu',\mu'',\nu')\sim I(z_{[i-1]})} \T_{\cost_i}(\mu',\mu'')\right)^{1/p}
+
\left(\Ex_{z \sim \pi} \sum_{i \in [n]} \Ex_{(\mu',\mu'',\nu')\sim I(z_{[i-1]})} \T_{\cost_i}(\mu'',\nu')\right)^{1/p}$$
The first term is exactly the first term on the right hand side of the inequality of the lemma. Therefore, all we have to do is to prove that
$$ \Ex_{z \sim \pi} \sum_{i \in [n]} \Ex_{(\mu',\mu'',\nu')\sim I(z_{[i-1]})} \T_{\cost_i}(\mu'',\nu') \leq \Ex_{{z} \sim \pi} \sum_{i \in [n]}  \T_{\cost_i}(\sigma_{z_{[i-1]}}).$$
In fact, we prove this statement for \emph{every} choice of $z$ and $i$, so ignoring $z,i$ we prove the claim:
$$  \Ex_{(\mu'',\nu')\sim I} \T_{\cost_i}(\mu'',\nu')  { \leq \Ex_{(\mu'',\nu') \sim I}   \T_{\cost_i}((\sigma^{-1}|\nu')^{-1}) =} \T_{\cost_i}(\sigma),$$
where the middle term is added for the proof.
        
        We now prove both the inequality and the equality above through the steps below.
        \begin{itemize}
            \item {\em Equality:} Since the average of $\nu' \sim J$ is $\nu_i$ and $\sigma^{-1}$ is a transport from $\nu_i$ to $\mu_i$, if we define $\cost'_i(y_i,x_i) = \cost_i(x_i,y_i)$, then by Part \ref{part:2} of Proposition \ref{prop:props-of-compose} we have
            $$\Ex_{(\mu'',\nu') \sim I}   \T_{\cost_i}((\sigma^{-1}|\nu')^{-1})  =  \Ex_{(\mu'',\nu') \sim I}   \T_{\cost'_i}(\sigma^{-1}|\nu')=\T_{\cost'_i}(\sigma^{-1}) = \T_{\cost_i}(\sigma).$$
            \item {\em Inequality:}  Again, using $\cost'_i(y_i,x_i) = \cost_i(x_i,y_i)$, we have 
            $$\Ex_{(\mu'',\nu') \sim I}   \T_{\cost_i}(\mu'',\nu') 
            = \Ex_{(\mu'',\nu') \sim I}   \T_{\cost'_i}(\nu',\mu'')
            \leq 
            \Ex_{(\mu'',\nu') \sim I}   \T_{\cost'_i}(\sigma^{-1}|\nu')
            = \Ex_{(\mu'',\nu') \sim I}   \T_{\cost_i}((\sigma^{-1}|\nu')^{-1}),           $$
            where the inequality is due to   the fact that $\T_{\cost'_i}(\nu',\mu'')$ is the optimal cost.
        \end{itemize}
\end{proof}

%% file: Results.tex
\section{Algorithmic Transport for Products} \label{sec:results}
In this section, we put together the tools from previous sections to derive algorithmic results about online transport for the setting that one of the source or target distributions is product. We then derive a corollary for the Gaussian measure.
We first define sequential samplers.

\begin{definition}[Sequential Sampler] \label{def:condSamp}  For a distribution $\nu$ in dimension $n$ with marginals $\nu_1,\dots,\nu_n$, we call $\hat{\nu}$ its \emph{sequential sampler} for $\nu$, if for \emph{all}  $y_{[i-1]}\sim \nu_{[i-1]}$ calling $\hat{\nu}(y_{[i-1]})$ returns an independent answer  $\hat{\nu}(y_{[i-1]}) \sim  \nu_i|y_{[i-1]}$. For queries $y_{[i-1]}\nin \Supp(\nu_{[i-1]})$, calling $\hat{\nu}(y_{[i-1]})$ returns $\bot$. We also  assign a (sequential sampling) cost $\oracleC_\nu(y_{[i-1]})$ to   query $y_{[i-1]}$, and call $\oracleC_\nu=\Ex_{y} \sum_{i \in [n-1]} \oracleC_\nu(y_{[i-1]})$ the \emph{average (sequential sampling) cost}  of $\hat{\nu}$. For an oracle-algorithm $A$ calling (a potentially randomized) set $\cQ$ of queries to $\hat{\nu}$, we define its \emph{average total cost} of calling $\hat{\nu}$ as $\oracleC_\nu^A = \Ex_{\cQ} \sum_{a \in \cQ} \oracleC_\nu(a).$\footnote{Since $\oracleC_\nu(y_{[i-1]})$ naturally measures the (e.g., computational) cost of sampling a coordinate conditioned on previously sampled coordinates, for natural settings and independent $\nu_1,\nu_2$, the value of  $\oracleC_{\nu}(y_{[1]})$ will be independent of $y_{[1]}$.}
\end{definition}

One natural way of using $\oracleC$ is to model sampling time, but it can model other costs as well.
The average cost $\oracleC_\nu$ of $\hat{\nu}$ is indeed the average total cost of the following simple algorithm $A$ that uses $\oracleC_\nu$ sequentially to obtain a full sample: Let $y_{[0]}$ be the empty string, and for $i \in [n]$, $A$ let $y_i = \hat{\nu}(y_{[i-1]})$. Also, when $\mu$ is a product distribution, then $\hat{\mu}$ is nothing other than a direct way of sampling from independent distributions $\nu_i$ for all $i \in [n]$.

\iffull
Before stating our main result, we define a notation for  transport cost to empirical sets.
\begin{definition}[Transports to Empirical] \label{def:emp}
For distributions $\mu$ and symmetric cost $\cost$,   let
$\empT_{\cost,k}(\mu)=\Ex_{\cX \sim \mu^k} \T_{\cost}(U_\cX,\mu)$ denote the cost of transporting $\mu$ to an empirical set of size $k$, where $ U_\cX$ is the uniform distribution over the multi-set $\cX$.
\end{definition}
\fi

\begin{theorem} [Main Result] \label{thm:main}
    Suppose $\src = \mu_1 \otimes \dots  \otimes \mu_n$  and $\trg$ are   distributions  over $\R^n$, with sequential samplers $\hat{\mu},
\hat{\nu}$ and corresponding oracle  cost functions $\oracleC_\mu,\oracleC_\nu$. Suppose   the transportation cost function $\cost$ is a metric (i.e., symmetric and satisfies the triangle inequality) and    $\cost^p$ is linear  over symmetric costs $\cost_1,\dots,\cost_n$.\footnote{An example is $\cost = \ell_p$.} 
    Then, there is an algorithm $A_k$, parameterized by $k$, that uses oracle access to samplers $\hat{\mu},\hat{\nu}$ and achieves the following:
    \begin{enumerate}
        \item $A_k^{\hat{\mu},\hat{\nu}}$ transports $\src$ to $\trg$ through an online coupling in time $\poly(n k)$ with $p$-cost\footnote{See Definition~\ref{def:transport}.}  
        $$ \T^{1/p}_{\cost^p}(A_k^{\hat{\mu},\hat{\nu}}) \leq   \delta + \lowerprod   $$
        %
        in which   $\delta =  2 \left( \sum_{i \in [n]} \empT_{\cost_i,k}(\mu_i)\right)^{1/p}$ and $\lowerprod = \lowerprod^{1/p}_{\cost^p}(\src,\trg)$ as in Definition~\ref{def:lower-bound}.\footnote{By Theorem~\ref{thm:prod}, $\Delta_{\cost^p}$ is also equal to $ \T^\OnT_{\cost^p}(\src,\trg) = \T^\OnC_{\cost^p}(\src,\trg)= \T^\OnG_{\cost^p}(\src,\trg)$.}

        \item The average total cost of $A$ calling $\hat{\mu},\hat{\nu}$ is as follows. $\oracleC^A_\nu \leq k \cdot \oracleC_\nu$ and   $\oracleC^A_\mu \leq k \cdot \oracleC_\mu$.\footnote{Note that because $\mu$ is a product distribution, if  $\oracleC_\mu$ models the computational cost of sampling from $\mu$, then we would have $\oracleC_\mu = \sum_{i \in [n]} \oracleC_{\mu_i}$, where $\oracleC_{\mu_i}$  models the computational cost of sampling from $\mu_i$.}
    
        \item There is an algorithm $B$ that achieves the same as $A$ does, but it transports $\trg$ back   to $\src$.
    \end{enumerate}


 \end{theorem}

\begin{proof} 
    At a high level, we use an empirical variant of the greedy algorithm (which is related to the KR transport) to design the algorithm.  The algorithm itself is quite simple; the bulk of the work goes into its analysis, which is quite delicate and uses many tools from Section~\ref{sec:tools}.

     \parag{The Transportation Algorithm $A$.} The algorithm $A$ works in $n$ rounds.          In round $i \in [n]$,  given $x_i \sim \mu_i$ find $y_i \sim \nu_i | y_{[i-1]}$ as described below. 
     \begin{enumerate}
        \item For $j \in [k]$, let $y^{(j)}_i \sim \hat{\nu}(y_{[i-1]})$ be   independent samples forming the multi-set $\cY$ of size $k$.
        \item Pick $t \gets [k]$ at random. For all $j \in [k], j \neq t$, let $x^{(j)}_i \sim \mu_i$ be $k-1$ independent samples. Additionally, let $x^{(t)}_i = x_i$ and  $\cX$ be the multi-set $\set{x_i^{(j)} \mid j \in [k]}$ of size $k$.
        \item Find an optimal transport between the two distributions $U_\cX,U_\cY$ under the cost $\cost_i$ (e.g., using the Hungarian method\footnote{This method can be implemented faster when the cost function is convex, in which case simply sorting $\cX,\cY$ gives us the optimal matching, via a monotone mapping.}) that is in the form of a matching between $\cX$ and $\cY$.\footnote{This can be proved, e.g., using the Birkhoff–von Neumann decomposition of doubly stochastic matrices.}
        \item Output $y_i \in \cY$ that is matched with $x^{(t)}_i = x_i \in \cX$.  
     \end{enumerate}

     
     We now analyze  the algorithm   $A$ above.


   \parag{Transportation:} $A$'s running time is clearly $\poly(kn)$. 
        We now prove that $A$'s algorithm produces an \emph{online} coupling between $\mu,\nu$, by showing that in round $i$, it couples $\mu_i$ and $\nu_i|y_{[i-1]}$. It is simple to check that all the elements of $\cX$ are distributed as $\mu_i$ and all the elements of $\cY$ are distributed as in $\nu_i|y_{[i-1]}$. At first, it might not be clear why $y_i$ is distributed as $\nu_i|y_{[i-1]}$, because the matching algorithm might change its distribution by picking it adversarially. However, since the algorithm hides the index of $x_i$ and statistically hides it among $\cX$, the final ``matched pair'' $(x_i,y_i)$ is a random edge of the optimal matching/transport. Therefore, $y_i$ is also distributed accurately, and hence $A$ is producing an online coupling.

        More formally, we can choose $t \in [k]$ at random \emph{after} the matching between $\cX,\cY$ is chosen. Moreover, the marginal distribution of $y_i^{(j)}$ is $\hat{\nu}(y_{[i-1]})$. Therefore, for every (even fixed) matching between $\cX,\cY$, picking $t$ at random will lead to picking $y_i = y_i^{(j)}$ where $j$ is the index of the sample in $\cY$ that is matched with the index $t$ in $\cY$. Therefore, $y_i \sim \hat{\nu}(y_{[i-1]})$.

    \parag{The Cost:} To analyze the transportation cost we apply Lemma~\ref{lem:Multi-Round-Intermediate} from Section~\ref{sec:tools}, which is stated in a more general form to better demonstrate the key ideas. 
    For a more modular proof, We first describe an important special case of Lemma~\ref{lem:Multi-Round-Intermediate} which covers the setting of interest in the proof.
\begin{lemma}[Special Case of Lemma~\ref{lem:Multi-Round-Intermediate}] \label{lem:special-case}
Add conditions below to the setting of Lemma~\ref{lem:Multi-Round-Intermediate}.
\begin{enumerate}
    \item \label{cond:1} The transport $\sigma_{z_{[i-1]}}$ from $\mu_i$ to $\nu_i|y_{[i-1]}$ is  assumed to be an \emph{optimal} transport.
    \item \label{cond:2} There is a parameter $k$ that determines how $\mu',\nu'$ are   sampled: $J(z_{[i-1]})$ first picks two random sequences $\cX \sim (\mu_i|z_{[i-1]})^k,\cY \sim (\nu_i|z_{[i-1]})^k$ of size $k$, and lets $\mu' = \cU_{\cX},\nu'=\cU_\cY$.
    \item \label{cond:3} The costs $\cost_i$  are symmetric (making $\cost$ a metric as well). 
\end{enumerate}
Then, we have
$$             \T^{1/p}_{\cost^p}(\pi) 
\leq  
2 \left( \Ex_{z \sim \pi} \sum_{i \in [n]} \empT_{k,\cost_i}(\mu_i|z_{[i-1]})\right)^{1/p}
            + \lowerbound^{1/p}_{\cost^p}(\pi).$$
\end{lemma}
\begin{proof}
    The last term simplifies to $\lowerbound^{1/p}_{\cost^p}(\pi)$ because of Condition~\ref{cond:1}.
    Assuming Condition~\ref{cond:2}, by Part~\ref{part:3} of Proposition~\ref{prop:props-of-compose} we have that $\sigma^{-1}_{z_{[i-1]}} \push \nu'_i$,  is distributed as $\cU_{\cX'}$ for an independently sampled $\cX' \sim (\mu_i|z_{[i-1]})^k$.
    Hence, by the triangle inequality of Proposition~\ref{prop:triangle} applied to triple $(\mu'_i=\cU_\cX,\mu_i,\nu'_i = \cU_\cY)$, we have the following
    \begin{align*}
        & \left(\Ex_{{z} \sim \pi} \sum_{i \in [n]} \Ex_{(\mu'_i,\nu'_i) \sim J(z_{[i-1]})} \T_{\cost_i}\left(  
            \mu'_i, \sigma^{-1}_{z_{[i-1]}} \push \nu'_i
            \right)  \right)^{1/p} \leq \\
        & \left( \Ex_{z \sim \pi} \sum_{i \in [n]} \Ex_{\cX \sim (\mu_i|z_{[i-1]})^k} \T_{\cost_i}(\cU_\cX,\mu_i|z_{[i-1]})\right)^{1/p} + \left( \Ex_{z \sim \pi} \sum_{i \in [n]} \Ex_{\cX \sim (\mu_i|z_{[i-1]})^k} \T_{\cost_i}(\mu_i|z_{[i-1]},\cU_\cX)\right)^{1/p},
    \end{align*}
    which by the symmetry of $\cost$ simplifies to $2 \left( \Ex_{z \sim \pi} \sum_{i \in [n]} \empT_{k,\cost_i}(\mu_i|z_{[i-1]})\right)^{1/p}$.
    
\end{proof}

To apply Lemma \ref{lem:special-case}, we make several observations. Our algorithm $A$ also uses empirical distributions $\cU_\cX, \cU_\cY$, and the costs $\cost_i$ are indeed symmetric. Hence, we can apply Lemma \ref{lem:special-case} to analyze the cost of $A$. Here we make several observations:
\begin{itemize}
    \item Since $\mu$ is product and that $A$ is online, the first term in Lemma \ref{lem:special-case} simplifies as
    $$  \Ex_{z \sim \pi} \sum_{i \in [n]} \empT_{k,\cost_i}(\mu_i|z_{[i-1]}) = \sum_{i \in [n]} \empT_{k,\cost_i}(\mu_i),$$
    which leads to the term $2\delta$ in our theorem statement.
    \item Since $A$ is producing an \emph{online} coupling  the second term in Lemma \ref{lem:special-case} simplifies to $\lowerprod^{1/p}_{\cost^p}(\mu,\nu) = \lowerbound^{1/p}_{\cost^p}(\pi_A)$, due to Part~\ref{prop:3} of Proposition~\ref{prop:props} and that $\mu$ is a product.
\end{itemize}
The simplified bounds establish the claimed cost of algorithm $A$.

   \parag{Oracle Costs:} In each round, $A$ asks $k-1$   samples from $\mu_i$ and $k$   samples from $\nu_i | y_{[i-1]}$. Furthermore, the previous samples $y_{[i-1]}$ are sampled according to $\nu_{[i-1]}$ itself, so the average total cost will be as claimed.

    \parag{Inverse Transport:} The reverse mapping uses the same algorithm for one dimension transport, but it maps $\nu_i | y_{[i-1]}$ to $\mu_i$, and inspection shows its transportation and (expected) total oracle costs will be the same as that of $A$. 
 \end{proof}

 \begin{remark}[Using CDFs Instead of Samplers]
An inspection of the proof of Theorem~\ref{thm:main}  reveals that the term $\delta$ in Theorem~\ref{thm:main} appears due to the fact that we work with empirical samples of $\mu_i$ and $\nu_i | y_{[i-1]}$. Once we have an effectively computable CDF for \emph{only one} of these distributions (i.e., for $\mu_i$ for all $i\in[n]$ or $\nu_i | y_{[i-1]}$ for all $i\in[n],y_{[i-1]}$) then this ``empirical term'' $\delta$ can be bounded by  half of what is stated. In particular, $A$ can now use the  the CDFs of the source or target distribution in round $i$ to directly couple the empirical samples with the non-empirical distribution. If we know both CDFs, then the term $\delta$ can be eliminated completely, achieving a  mapping that is optimal among online transports, which is the same as the KR transport for convex cost functions $c_i$.
%


\end{remark}

\subsection{Extension to Oracle-Set Transports}
In this subsection, we study how to use the main result of Theorem~\ref{thm:main} and obtain transports from the same $\mu$ to a richer set of distributions that can be obtained from $\nu$ by conditioning $\nu$ on an event $\cS$ of not-so-small probability. This will be especially useful later when we focus on transporting Gaussian distributions to the same distributions conditioned on an event $\cS$. 
To prove this extension, we prove a general result about using sequential samplers for $\nu$ to obtain sequential samplers for $\nu|\cS$.

\begin{theorem}[Sequential Samplers for Event-conditioned Distributions] \label{thm:nested}
Suppose $\nu$ is an $n$-dimensional distribution that has a sequential sampler $\hat{\nu}$ with average cost $\oracleC_\nu$. Suppose $\cS$ is an event of measure $\nu(\cS) \geq \eps$, and $\omega=\nu|\cS$ is $\nu$ conditioned on $\cS$. Then, there is an algorithm $O$ that uses oracle $\hat{\nu}$ and a membership oracle  $\cS$ and achieves the following.
\begin{enumerate}
    \item \label{Part:SS:1} For all $y_{[i-1]} \sim \omega_{[i-1]}$, $O^{\cS,\hat{\nu}}(y_{[i-1]}) \sim \hat{\omega}(y_{[i-1]})$.
    \item \label{Part:SS:2} If we define $\oracleC_\omega(y_{[i-1]})$ be the average total cost of $O^{\cS,\hat{\nu}}(y_{[i-1]})$ querying $\hat{\nu}$, and if we define ${\oracleC}_{\nu}(i)=\Ex_{y \sim \nu} \oracleC_\nu(y_{[i-1]})$,   then  $$\oracleC_\omega \leq \frac{1}{\eps}\sum_{i \in [n]} i\cdot {\oracleC}_{\nu}(i) \leq n \cdot \frac{\oracleC_\nu}{\eps}.$$ 
    \item \label{Part:SS:3} When  iteratively sampling   $(y_1,\dots,y_n) \sim  \omega$, the expected number of calls to $\cS$ in round $i$ is at most $1/\eps$, making the total expected number of calls to $\cS$ to be at most $n/\eps$. 
    \item \label{Part:SS:4} The running time of the iterative sampling of $(y_1,\dots,y_n) \sim  \omega$, relative to the provided oracles $\hat{\nu},\cS$ is at most $O(n^2/\eps)$.
\end{enumerate}
\end{theorem}

In other words, one can use $O^{\cS,\hat{\nu}}$ to emulate a sequential sampler for $\omega=\nu|\cS$ in such a way that the average cost of obtaining a full sequence $y \sim \nu|\cS$ using $n$ nested calls to the provided sequential sampler only goes up (at most) by a multiplicative factor $n/\Pr[\cS]$.

The main idea in the proof is to use rejection sampling with a subtle analysis. Namely, $O^{\cS,\hat{\nu}}$ simply keeps   using $\hat{\nu}$ to obtain full sequences  multiple times until the sample sequence falls within the event $\cS$.  The full proof follows.

\begin{proof} [Proof of Theorem~\ref{thm:nested}]
 For  $v=(v_1,\dots,v_n)$, let $v_{\geq i} = (v_i,\dots,v_n)$ and $v = (v_{[i-1]},v_{\geq i})$.

Our algorithm $O^{\cS,\hat{\nu}}(y_{[i-1]})$ samples from $\hat{\omega}(y_{[i-1]})$ as follows. 
\begin{enumerate}
    \item Sample from $\nu|y_{[i-1]}$ as follows: for $j = i, \dots, n$ sample fresh values $ y_{j} \sim \hat{\nu}(y_{[j-1]})$.
    \item If $y =   (y_{[i-1]},y_{\geq i}) \in \cS$, then output $y_i$; otherwise, go back to the previous step.
\end{enumerate}
We refer to each execution of the two steps above (that has exactly one call to $\cS$) a \emph{trial}. 

Part \ref{Part:SS:1}  follows from the fact that the above sampling process is a simple rejection sampling.
To prove Part \ref{Part:SS:2}, let $H(y_{[i-1]})$ be a random variable that counts the number of trials, and let its expectation be
$$h(y_{[i-1]}) = \Ex  [H(y_{[i-1]})] = \frac{1}{\Pr_{y \sim \nu|y_{[i-1]}}[y \in \cS]}.$$ 
Also let $\ol{\oracleC}_\nu(y_{[i-1]}) = \Ex_{y \sim \nu|y_{[i-1]}} \sum_{j \geq i} \oracleC_\nu(y_{j-1}).$ It holds that $\Ex_{y \sim \nu} \ol{\oracleC}_\nu(y_{[i-1]}) = \sum_{j\geq i} {\oracleC}_{\nu}(i)$. Using these notations, the oracle sampling cost of $\hat{\omega}(\cdot)$ at $y_{[i-1]}$ will  be
$$\oracleC_\omega(y_{[i-1]}) = h(y_{[i-1]}) \cdot \ol{\oracleC}_\nu(y_{[i-1]}).$$
Therefore, the average cost of $\hat{\omega}$  will be
$$ \oracleC_\omega =  \Ex_{y \sim \omega} \sum_{ i \in [n]} h(y_{[i-1]}) \cdot  \oracleC_\omega(y_{[i-1]}) =  \sum_{ i \in [n]} \Ex_{y \sim \omega} h(y_{[i-1]}) \cdot \oracleC_\omega(y_{[i-1]}).$$


A subtle point is that, in the above sums the first half $y_{[i-1]}$ is sampled conditioned on $\cS$, while the second half is done without such conditioning. We claim that for each $i$, we have 
\begin{equation} \label{eq:each-i}
  \Ex_{y \sim \omega} h(y_{[i-1]}) \cdot \oracleC_\omega(y_{[i-1]}) \leq \frac{1}{\eps} \cdot \Ex_{y \sim \nu}  \oracleC_\omega(y_{[i-1]}).
\end{equation}
Note that if  Eq. (\ref{eq:each-i}) holds, then we conclude Part \ref{Part:SS:2}, because we get:
$$\oracleC_\omega \leq   \sum_{ i \in [n]} \frac{1}{\eps} \Ex_{y \sim \nu}  \oracleC_\omega(y_{[i-1]}) = \frac{1}{\eps}\cdot \sum_{i \in [n]} \sum_{j \geq i} \oracleC_\nu(i) = \frac{1}{\eps} \sum_{i \in [n]} i\cdot {\oracleC}_{\nu}(i).$$


\remove{
$$\Ex_{y  \sim \omega } \sum_{i \in [n]} h(y_{[i-1]}) =
\sum_{i \in [n]} \Ex_{y_{[i-1]} \sim \nu[i-1]} h(y_{[i-1]})
= \sum_{i \in [n]} q_i, \text{~for~} q_i = \Ex_{y_{[i-1]} \sim \nu[i-1]} \frac{1}{\Pr_{y_{\geq i} \sim \mu_{\geq i}} [y \in \cS]}.$$
We now claim that $q_i \leq 1/\eps$  for all $i$, in which case we would get $\Ex[Q] \leq nk/\eps$. 
}

The following lemma proves Eq. (\ref{eq:each-i}) using $\cU = \Supp(\nu_{[i-1]}), \cV= \Supp(\nu_{\geq i})$,   $\sigma = \nu$, $f(y) = \ol{\oracleC}_\nu(y_{[i-1]})$ and   $\cS$ as before.
  
\begin{lemma}[Expected Cost of Two-Step Sequential Sampling] \label{lem:conditional-samp}
    Suppose  $\sigma$ is distributed over $\cU \times \cV$ with margins $\sigma_\cU, \sigma_\cV$, and $\cS \se \cU \times \cV$ has probability $\sigma(\cS) = \eps$. Also, suppose $f$ is a random variable defined over $\sigma$ with average $\bar{f}$.
    Consider the following process:  (1) Sample $u \sim \sigma_\cU|\cS$, which is the marginal distribution of $\cU$ in $\sigma|\cS$ and let $\eps_u = \Pr_{v \sim \sigma^{u}_\cV}[(u,v) \in \cS]$, in which $\sigma^u_\cV$ is the marginal distribution over $\cV$  in $\sigma$ conditioned on $\sigma_\cU=u$. (2)  Sample $v\sim \sigma^u_\cV|\cS$. Then, 
    $$ \Ex_{(u,v)} \frac{f(u,v)}{\eps_u} \leq \frac{\bar{f}}{\eps}.$$
    
\end{lemma}
\begin{proof}
    We write the proof for the discrete setting.  
    A similar proof holds in   general.
    For each $u \sim \sigma_\cU$, define $p_u = \Pr[\sigma_\cU = u]$    and $f_u = \Ex_{v \sim \sigma_\cV|u} f(u,v)$. 
    We have $\eps = \sum_{u \in \cU} p_u \eps_u$, and  $q_u = \frac{p_u\cdot \eps_u}{\eps}$ is the probability we sample $u$ in the sampling process of the lemma statement. Then, if we let $\cU_\cS = \set{u \mid \eps_u >0}=\Supp(\sigma_\cU|\cS)$, we have
    $$  \Ex_{u } \frac{1}{\eps_u} \Ex_{v|u} f(u,v)  = \sum_{u \in \cU_\cS} \frac{q_u}{\eps_u} f_u = \sum_{u \in \cU} \frac{p_u}{\eps} f_u \leq \sum_{u \in \cU_\cS} \frac{p_u}{\eps} f_u = \frac{ \bar{f}}{\eps}. $$
\end{proof}
To prove Part \ref{Part:SS:3}, using Lemma \ref{lem:conditional-samp} and $f(u,v)=1$, we conclude  that the expected number of times we call the $\cS$ oracle at each node $y_{[i-1]}$   is at most $1/\eps$.

To prove Part \ref{Part:SS:4} we can simply use a fake oracle sampling  cost of $\hat{\oracleC}'_\nu(\cdot)=1$. Then the claim about the running time follows  from Part \ref{Part:SS:2}.
\end{proof}

\parag{Deriving oracle-set transports.}  Using Theorem~\ref{thm:nested}, we can derive more transportation results from  Theorem~\ref{thm:main} by conditioning $\nu$ on an arbitrary event $\cS$ for which we have a membership oracle at hand. Note that the parameter $\Delta$ will change to a new value, but the key point is that we can control the cost of sequential samples from the new oracle, so long as we could do so for the initial oracle. 
Another interesting application of Theorem~\ref{thm:nested} is to transport a product distribution $\mu$ to $\mu|\cS$ for an arbitrary event $\cS$. We first define  such set-uniform transportation algorithms.

\begin{definition}[Oracle Set-Transport] \label{def:set-transport}
    For distribution $\mu$ and transportation cost $\cost$, we say that $(\mu,\cost)$ has a \emph{set-transport} of cost at most $\kappa(\cdot)$ for a non-increasing function $\kappa \colon [0,1] \To [0,1]$, if for every event $\cS \se \Supp(\mu)$, it holds that $\T_{\cost}(\mu,\mu|\cS) \leq \kappa(\cS).$ 
    We further say that $(\mu,\cost)$ has an \emph{oracle} set-transport of cost at most $\kappa(\cdot)$ if there is a single algorithm $A$ such that  with oracle   membership queries for an arbitrary set $\cS$ and sampling queries for  $\mu$,  $A^{\cS,\mu}$ produces a transport of cost at most $\kappa(\cS)$ from $\mu $ to $\mu | \cS$.
\end{definition}

\begin{corollary}[Oracle-Set Transports] \label{cor:main-for-prod}
Suppose the assumptions of Theorem~\ref{thm:main} hold. Then, 
\begin{enumerate}
    \item There is an oracle set transport  algorithm $M_k$ for $(\mu,\cost)$ such that, for all events $\cS$ defined over $\mu$, $M^{\cS,\hat{\mu}}_k$ transports $\mu$ to $\mu|\cS$ in expected time $\poly(nk/\eps)$ and $p$-cost $\T^{1/p}_{\cost^p}(M_k^{\hat{\mu},\hat{\nu}}) \leq   \delta + \lowerprod$, in which $\delta$ is as in Theorem~\ref{thm:main} and $\Delta = \lowerprod^{1/p}_{\cost^p}(\src,\src|\cS)$.
    \item There is an algorithm $N_k$ such that, for all events $\cS$ defined over {$\nu$}, $N^{\cS,\hat{\mu},\hat{\nu}}_k$ transports $\mu$ to {$\nu|\cS$} in expected time $\poly(nk/\eps)$ and $p$-cost $\T^{1/p}_{\cost^p}(N^{\cS,\hat{\mu},\hat{\nu}}_k) \leq   \delta + \lowerprod$, in which $\delta$ is as in Theorem~\ref{thm:main} and $\Delta = \lowerprod^{1/p}_{\cost^p}(\src,\nu|\cS)$. Moreover, $\oracleC_\nu^N \leq n \cdot \oracleC_\nu^A /\eps$, for $A$ of Theorem~\ref{thm:main}.
\end{enumerate}
In both cases above, the expected number of calls to $\cS$ is at most $kn/\eps$, and the transportation can be reversed  with the same upper bounds on the running time and oracle costs.
\end{corollary}

\section{Algorithmic Transport for   Gaussian}
In this section we focus on cases where at least one of the two distributions involved in the transport is Gaussian. We first use the main result of Theorem~\ref{thm:main} and derive an algorithmic variant of Talagrand's result \cite{talagrand1996transportation} about transporting Gaussian measure to arbitrary distributions with bounded KL divergence from Gaussian. We then derive, as a corollary, a computational concentration result for the Gaussian source measure under the $\ell_2$ distance. Finally, we focus on finding (optimal) online transports in cases where both the source and destination are Gaussians, but they could be arbitrary (non-product) Gaussians.



\subsection{Algorithmic Variant of Talagrand's Transport for Gaussian}
\begin{theorem}[Algorithmic Version of Talagrand's Gaussian Transport Theorem] \label{thm:genTal}
 Let $\Phi^n$ be the standard Gaussian in dimension $n$ and $\nu$ be an arbitrary distribution in $\R^n$.
There is an algorithm $A_k$,  with  integer parameter $k$,   such that whenever $A^{\hat{\nu}}_k$ is provided with a sequential sampler $\hat{\nu}$ for $\nu$, the following properties hold.
\begin{enumerate}
    \item For all $p \geq 1$ and $\nu$, $A_k^{\hat{\nu}}$ transports $\Phi^n$ to $\nu$ in time $O(n k \log k)$ with  $p$-cost    at most 
    $$\T^{1/p}_{\ell_p^p}  (A^{\hat{\nu}}_k) \leq \Delta^{1/p}_{\ell_p^p}(\Phi^n,\nu) + \left(O_p(n k^{-1/2})\right)^{1/p}.$$
    For $p =2$, by the Talagrand inequality of Theorem~\ref{cor:genTal}, we have $\Delta_{\ell_2^2}(\Phi^n,\nu) \leq \ {2 \KL{\Phi^n}{\nu}}$.
    \item The average total oracle cost of  $A_k^{\hat{\nu}}$    is at most  $k\cdot \Ex_{y \sim \nu} \sum_{i \in [n]} \oracleC(\nu_{i}|y_{[i-1]}).$
    \item There is an algorithm $B_k^{\hat{\nu}}$ that achieves the same as $A_k^{\hat{\nu}}$, but it transports $\nu$ back to $\Phi^n$. 
\end{enumerate}
\end{theorem}


\iffull
\begin{remark}[Cases of $p \in [1,2)$]
    Theorem~\ref{thm:genTal} only states the special case of $p=2$ when we apply Talagrand's result for Gaussian. However, as stated in Section~\ref{sec:Trans-Ent},  we can apply Corollary~\ref{cor:genTal} to obtain a result for all $p \in [1,2]$ as follows
    $\lowerprod_{\ell_p^p}(\Phi_n,\nu) \leq n^{1-p/2} \cdot (2\KL{\nu}{\Phi_n})^{p/2}$.
\end{remark}
\fi

\begin{remark}[Working with $\ell_p$ instead of $\ell_p^p$] One might wonder what happens if we want to measure (and upper bound) transfer costs using $\ell_p$ rather than $\ell_p^p$. However, this can be obtained using   Jensen's inequality (or rather the monotonicity of Wasserstein $p$-costs for a fixed cost $\cost$). In particular,  for every   coupling $\pi$, we have 
$\T_{\ell_p}(\pi) \leq   \T^{1/p}_{\ell^p_p}(\pi)$ for all   $p \geq 1$. Hence Theorem~\ref{thm:genTal} is stated in the stronger form already.
\end{remark}

\begin{proof} [Proof of Theorem~\ref{thm:genTal}]
  The proof follows directly from Theorem~\ref{thm:main} and Corollary~\ref{cor:empGaussian}. Namely, we use Corollary~\ref{cor:empGaussian} to bound the term $\delta$ in Theorem~\ref{thm:main} that upper bounds the transportation cost of empirical Gaussian from the Gaussian itself.
 One small point here is that, we will not need oracle samplers from the Gaussian itself, as we can use well-known sampling methods  such as the Box-Muller method that generate such samples efficiently \cite{paley1934fourier}.\footnote{In particular, given two independent and uniform $u_1,u_2 \sim [0,1]$, it holds that 
$v_1 = \sqrt{-2\ln u_1} \cos(2\pi u_2), v_2 = \sqrt{-2\ln u_1} \sin(2\pi u_2)$ are independent  samples $v_1,v_2 \sim \cN(0,1)$.}
\end{proof}

 We now focus on a special case of interest, in which the target distribution $\nu$ is $\Phi^n | \cS$ for an event $\cS$ of probability $\Phi^n(\cS) = \eps$, and show that in this case, one can have a single online transportation algorithm that uniformly works for all $\cS$  by merely accessing $\cS$ through a membership oracle. (See Definition \ref{def:set-transport}.)


\begin{theorem}[Oracle-Set Transport for Gaussian Measure]  \label{thm:Gaussian-One-Way}
Let $\Phi^n$ be the standard Gaussian in dimension $n$. There is an (online) oracle-set transport algorithm $A_k$   for $\Phi^n$ such that:
\begin{enumerate}
    
    \item For all $p \in [1,2]$ and $\cS$ of measure $\Phi^n(\cS) = \eps$,
    $$ \T^{1/p}_{\ell^p_p}(A^\cS_k) \leq  \kappa^{1/p}(\eps) = n^{1/p-1/2}\sqrt{2 \ln \nf{1}{\eps}}+ \left(O_p(n k^{-1/2})\right)^{1/p},$$
which  is at most 
    $  (1+\gamma) \cdot n^{1/p-1/2}\sqrt{2 \ln \nf{1}{\eps}}$, for sufficiently large  $k = \poly(n,1/\eps,1/\gamma)$.
 \item   In expectation,  $A_k^\cS$ asks  at most $k n/\eps$  queries to $\cS$ and runs in   $\poly(nk/\eps)$.
 \item  There is an algorithm $B_k$ that achieves the same, but $B^\cS_k$   transports $\Phi^n|\cS$ back to $\Phi^n$.

 \end{enumerate}
\end{theorem}

 \begin{proof}[Proof of Theorem~\ref{thm:Gaussian-One-Way}]
To prove Theorem~\ref{thm:Gaussian-One-Way} we first use the first item of Corollary~\ref{cor:main-for-prod} where $\mu=\Phi^n$. This way, we already know that the running time of the transportation algorithm and its number of calls to $\cS$ are bounded as stated.

Then, we need to bound both terms $\Delta,\delta$. To bound $\delta$, we again use Corollary~\ref{cor:empGaussian} as we did in the proof of Theorem~\ref{thm:genTal}. To bound $\Delta$, we again use Corollary~\ref{cor:genTal}  and the well-known fact that $\KL{\mu|\cS}{\mu}\leq \ln 1/\eps$ for $\cS$ such that $\mu(\cS) \geq \eps$ (applied to $\mu=\Phi^n$). 
\end{proof}

Due to our transports being ``reversible'', one can obtain a variant of the result above that transports conditional distributions to conditional distributions through composition. Specifically, we show how to algorithmically transport any two measures that are KL-close to the Gaussian measure. To prove this, we use a crucial aspect of our results above, where we stated that the transports can be pushed  \emph{back} to the Gaussian measure as well. 

\begin{corollary}[Oracle Set-to-Set Transportation for Gaussian Measure]  \label{cor:Gaussian-Two-Way}
Let $\Phi^n$ be the standard Gaussian in dimension $n$. Suppose $C_k$ is an algorithm  that uses   membership queries to \emph{two}   sets $\cS,\cT$ and  simply composes running $B^\cS$ followed by running $A^\cT$ of Theorem~\ref{thm:Gaussian-One-Way}. 
Then, for all $p \in [1,2]$, $C_k^{\cS,\cT}$ transports $\Phi^n|\cS$ to $\Phi^n|\cT$ with cost
     $$ \T^{1/p}_{\ell_p^p}(C^{\cS,\cT}_k) \leq \kappa^{1/p}(\eps) + \kappa^{1/p}(\delta),$$
in which $\eps =\Phi^n(\cS),\delta = \Phi^n(\cT)$, $C_k^{\cS,\cT}$ and $\kappa(\cdot)$ is the   cost of the oracle-set transportation of Theorem~\ref{thm:Gaussian-One-Way} for cost $\ell_p^p$. For $p=2$,
$\T_{\ell_2}(C^{\cS,\cT}_k) \leq \T^{1/2}_{\ell_2^2}(C^{\cS,\cT}_k)\leq O(\sqrt{\ln \nf{1}{\eps}} + \sqrt{\ln \nf{1}{\delta}})$.
\end{corollary}

\begin{proof}
    The proof follows from Theorem~\ref{thm:Gaussian-One-Way} by first transporting $\Phi^n|\cS$ to $\Phi^n$ and then composing this (online) transport with a second (online) transport that goes from $\Phi^n$ to $\Phi^n|\cT$. Then, by Lemma~\ref{lem:onlineComposes} this composition is both online and has the stated transportation cost.
\end{proof}

%% file: CoM.tex
\subsection{Dimension-Independent Computational Concentration for Gaussian} \label{sec:CoM}
It is well-known that transportation inequalities can be used to derive concentration of measure results \cite{gozlan2010transport}. Recently, a \emph{computational} variant of this phenomenon has been explored~\cite{mahloujifar2019can,etesami2020computational}, which bears similarities to how we make transportation algorithmic. In a computational concentration result, we need an algorithm that maps ``most'' of the sampled points from the space to any ``sufficiently large'' event $\cS$, algorithmically. The   ``cost'' of the concentration is (a worst-case)  allowed distance $d$ that the algorithm is allowed to move the points, and its error is the fraction of the sampled points that it fails to map to $\cS$ withing the allowed distance $d$.
The work of \cite{etesami2020computational} obtained such results optimally for some settings (e.g., Gaussian  under  $\ell_1$ distance), however they left open obtaining  an optimal (dimension-free) computational concentration result for the Gaussian space under the $\ell_2$ distance.

Using Theorem \ref{thm:Gaussian-One-Way}, we can resolve the question left open in \cite{etesami2020computational} and derive such  optimal  {computational} concentration  for the Gaussian space under $\ell_2$ as a simple corollary to our algorithmic transport result.  
Theorem \ref{cor:CoM} below follows from Theorem \ref{thm:Gaussian-One-Way} and the Markov inequality. Using $p=2$  below implies the desired dimension-independent result.

\begin{corollary}[Computational Concentration for Gaussian] \label{cor:CoM}
For all $\eps,\delta,\lambda,p \in [1,2]$, given oracle access to $\cS\se \R^n$, $A^\cS_k(x)$ of Theorem \ref{thm:Gaussian-One-Way} runs in $\poly(\frac{n}{\eps \lambda})$-time and with probability $1-\delta$ over  $x \sim \Phi^n$, it finds a point $y \in \cS$ of distance  
$$\ell_p(x,y) \leq \frac{ (1+\gamma) \cdot n^{1/p-1/2}\sqrt{2 \ln \nf{1}{\eps}} }{\delta}.$$
\end{corollary}

\remove{
and arbitrary constant $\delta$, the algorithm $A^{\cS}_k$ achieves $\ell_2(x,y) \leq O(\sqrt{\ln \nf{1}{\eps}})$ for $1-\delta$ fraction of $x \sim \Phi^n$. To derive Corollary \ref{cor:CoM} from  Theorem \ref{thm:Gaussian-One-Way},  recall that
$$\Ex_{x \sim \Phi^n, y=A^{\cS}(x)} \ell_p(x,y) = \T_{\ell_2}(A^{\cS}) \leq \T^{1/2}_{\ell_2^2} (A^{\cS}) \leq  (1+\gamma) \cdot n^{1/p-1/2}\sqrt{2 \ln \nf{1}{\eps}}.$$
Then, the claim follows by an application of the Markov inequality.
}

%% file: TechTools.tex
\section{Remarks about Online Transport} \label{sec:properties}

If one does not care about an online transport to be done algorithmically, the following proposition characterizes online transports information theoretically.
\begin{proposition}[Characterizing Online Transport]  \label{prop:online}
    Suppose $\pi$ is a coupling  between $\mu,\nu $ and $(x,y) \sim \pi$. Then, $\pi$ is an online transport from $\mu$ to $\nu$ if and only if, for all $i \in [n]$ it holds that $x_i$ and $y_{[i-1]}$ are  independent conditioned on $x_{[i-1]}$.
\end{proposition}




\parag{On Symmetry.} Recall that  the optimal  transportation cost function is symmetric with respect to $\mu$ and $\nu$, whenever the cost function is symmetric.  However, as we will see (in the example after Theorem \ref{thm:prod}) this is not true for online transport.
 Yet, similarly to the offline transport,  when the cost $\cost$ is symmetric, we always have the symmetry condition  $\T^\OnC_\cost(\mu,\nu) = \T^\OnC_\cost(\nu,\mu)$.

  $\T^\OnC_{\cost}(\mu,\nu)$ \emph{upper bounds}  $\T^\OnT_{\cost}(\mu,\nu)$ (and $\T^\OnT_{\cost}(\nu,\mu)$ when the cost is symmetric). To understand the difference between online coupling and (one-way) online transports, it is  natural to ask the conditions under which $\T^\OnC_{\cost}(\mu,\nu)=\max \set {\T^\OnT_{\cost}(\mu,\nu),\T^\OnT_{\cost}(\nu,\mu)}$ for symmetric $\cost$. As we will see   in Theorem \ref{thm:prod}, this equality  sometimes holds under natural conditions on $\cost,\mu,\nu$.


 The following lemma can be derived from a simple application of linearity of expectation.


\begin{remark}[Greedy couplings could be sub-optimal online transports.] A natural question is whether optimal online transports can always be obtained through some  greedy couplings. The following example shows that this is not always the case. Let $\mu,\nu$ both be uniform over $\set{(0,0),(1,1)}$. The cost over the first coordinate $\cost_1$  is the Hamming cost and $\cost_2(x_2,y_2)=2(1-\cost_1(x_2,y_2))$ is the opposite of the Hamming cost multiplied by two. The final cost function is linear $\cost((x_1,x_2),(y_1,y_2)) = \cost_1(x_1,y_1)+\cost_2(x_2,y_2)$. Then, the greedy algorithm $G$ will couple $x_0=x_1$ and gets $\cost(G) = 2$, while $\T^\OnC_\cost(\mu,\nu)=1$ in which we   have $x_2=y_2$.
\end{remark}




\remove{
    $(\src,\trg)$ and  cost  function $\cost$ that is linear over $\cost_1,\dots,\cost_n$, we define their \emph{Lambda cost function} as,
$$\lowerbound_{\cost}(\src,\trg) = \left(\Ex_{(y_1,\dots,y_n)\sim \trg}\sum_{i\in[n]}\T_{\cost_i}(\mu_i,(\nu_i| y_1,\dots,y_{i-1}))\right)^{1/p}  \text{and}~~~ \lowerbound_{\cost}(\mu,\nu) = \lowerbound_{1,\cost}(\mu,\nu) .$$
}

\begin{remark}
    [$\T^\OnT_\cost(\trg,\src)<\T^\OnC_\cost(\src,\trg)$ could happen for a product $\src$ and symmetric $\cost$] 
Since being online is a stronger condition, we always have 
$\T^\OnT_{\cost}(\trg,\src) \leq \T^\OnC_{\cost}(\src,\trg),$ whenever the cost function is symmetric. However, it  turns out that $\T^\OnT_{\cost}(\trg,\src) < \T^\OnC_{\cost}(\src,\trg)$ can happen even for product $\src$ and symmetric $\cost$.  Let $\mu_1,\mu_2$ be uniform over $\bits$ and independent. $\nu_1$ is uniform over \emph{two} bits $\bits^2$, and $\nu_2$ is always equal to the \emph{second} bit of $\nu_1$ (hence $\nu$ is not product). Also, for $x_1 \in \bits$ and $y_1 \in \bits^2$, let $\cost_1(x_1,y_1)$ be the Hamming distance between $x_1$ and the \emph{first} bit of $y_1$, and for $x_2,y_2 \in \bits$, $\cost_2(x_2,y_2)$ be their Hamming distance. Then, we have $\T^\OnT_\cost(\src,\trg)=\T^\OnC_\cost(\src,\trg) = 1/2$, while $\T^\OnT_\cost(\trg,\src) = 0$.
\end{remark}

\parag{Separating Online and Offline Transports.} Of course every online transport is also an offline one. So a natural question is how far these quantities can be for natural distributions and transportation costs. For a simple example, let $\mu,\nu$ be two distributions over $\bits^n$, for which $\ell_p^p$ becomes the same thing as the Hamming distance and we have $\T_{\H}(\mu,\nu) \leq n$.
We now use the characterization of Theorem \ref{thm:prod} to show that even when $\mu$ is product and the cost function is Hamming, $\T^\OnT_{\cost}{(\mu,\nu)}$ could be $n$ times bigger than $\T_{\cost}(\mu,\nu)$.  

\begin{claim}
    For distributions $\mu,\nu$ below, we have $\T^\OnT_{\H}(\mu,\nu) = n \cdot \T_{\H}
    (\mu,\nu).$
    \begin{itemize}
    \item Let $\mu$ be the uniform distribution over $\set{0,1}^n$.
    \item Let $x=0^n,x'=10^{n-1} \in \set{0,1}^n$. Then, for $\eps \leq 2^{-n}$, define the distribution $\nu$ as follows.
    \begin{itemize}
        \item $\nu(x) = 2^{-n}+\eps$.
        \item $\nu(x') = 2^{-n}-\eps$.
        \item If $y \neq x,x'$, then $\nu(y) = \mu(y) = 2^{-n}$.
    \end{itemize}
\end{itemize}
\end{claim}
\begin{proof}
     We have $\T_{\H}(\mu,\nu) = \eps$ by simply moving mass $\eps$ from $x$ to $x'$. 
Computing the online cost function is more tricky. 
By Theorem \ref{thm:prod}, the linearity of expectations, and that Hamming distance in dimension one is the same as the total variation distance. Therefore, we have 
$$ \T^\OnT_{\H}(\mu,\nu) = \sum_i \Ex_{y_{[i-1]}\sim \nu_{[i-1]}} \TV(U_1, \nu_i|y_{[i-1]}),$$
where $U_1$ is the uniform one-bit distribution.
This is equivalent to sampling $y_{[i]} \sim \nu_{[i]}$ and changing the last bit $y_i$ so that the $i$th bit is now a uniform bit for every $y_{[i-1]}\sim \nu_{[i-1]}$. The optimal strategy for this goal changes the $i$th bit exactly with probability $\eps$ as follows.
\begin{itemize}
    \item For $i=1$, we move $\eps$ measure of the sample $1$ to $1$.
    \item For $i>1$, we move $\eps/2$ measure of $0^i$ to $0^{i-1}1$ and take $\eps/2$ measure For $10^{i-1}$ from $10^{i-2}1$. This leads to changing the $i$th bit also with probability $\eps$.
\end{itemize}
Therefore, we have $\T^\OnT_{\H}(\mu,\nu)=n\eps$.
\end{proof}

%% file: Reductions.tex
\section{Reductions for Algorithmic Transport}
In this section, we introduce a  notion of reductions that is useful for deriving algorithmic and online transport. We then use our definition to derive new algorithmic and online transports.

\begin{definition}[Transport Reductions]\label{def:reduction}
    Suppose $M_1=(\mu_1,\cost_1),M_2=(\mu_2,\cost_2)$ be two distribution-cost pairs. We say that there is an    $\alpha$-reduction $R$ from  (oracle-set transports in) $M_1$ to   (oracle-set transports in)  in $M_2$ if the following hold: 
    \begin{enumerate}
        \item {\em Transport Mappings:} There are mappings $f,g$ such that the following hold.
        \begin{itemize}
            \item $f \colon \Supp(\mu_1) \To \Supp(\mu_2)$ is a  randomized mapping such that $f(x) \sim \mu_2$ whenever $x \sim \mu_1$. I.e., the coupling  $(x,f(x)), x\sim \mu_1$ is a Kantorovich transport from $\mu_1$ to $\mu_2$.
            \item $g \colon \Supp(\mu_2) \To \Supp(\mu_1)$  is a deterministic mapping such that $g(y) \sim \mu_1$ whenever $y \sim \mu_2$. I.e., the coupling $(y,g(y)), y \sim \mu_2$ is  a Monge transport from $\mu_2$ to $\mu_1$.
        \end{itemize}
        \item {\em Lipschitz Condition:} For all $x_1 \sim \mu_1, x_2 \sim f(x_1), x'_2 \sim \mu_2$, we have 
        $$\cost_1(x_1,x'_1=g(x'_2)) \leq \alpha \cdot \cost_2(x_2,x'_2).$$ 
        For example, this happens if (1) $g$ is the inverse of $f$, and (2) $g$ is $\alpha$-Lipchitz.
    \end{enumerate}
        We sometimes denote the reduction $R$ itself by its mappings $R=(f,g)$.
        
    Furthermore we define the following extra properties for any  reduction.
    \begin{enumerate}
        \item Algorithmic Aspect: The reduction is in complexity class $\cA$, if   $f,g \in \cA$.
        \item Online Aspect: The reduction is online, if  $f$ and $g$ are online transports.
    \end{enumerate}
\end{definition}

The following lemma states that reductions compose, and it has a straightforward proof.

\begin{lemma}[Composition of Transport Reductions] \label{lem:compose-Reductions}
Suppose there is an   $\alpha$-reduction   from   transports in $M_1$ to   transports  in $M_2$, and there is a $\beta$-reduction   from   transports in $M_2$ to   transports  in $M_3$. Then, there is a $\alpha\beta$-reduction from   transports in $M_1$ to   transports  in $M_3$. Furthermore, if the first two reductions are online, so is their composition, and the complexity of the composed reduction is bound by running the first two reductions sequentially.
\end{lemma}

The following theorem shows that reducing transports in distribution-cost space $M_1$ to $M_2$ indeed implies in the intuitive anticipation that together with a transportation result in $M_2$, it implies a transpiration result in $M_1$.
\begin{theorem} \label{thm:using-reductions}
    For distribution-cost pairs $M_1=(\mu_1,\cost_1),M_2=(\mu_2,\cost_2)$, suppose we have:
    \begin{enumerate}
        \item There is a $\alpha$-reduction from transports in $M_1$ to transports in $M_2$.
        \item There is a set-transport   for $\mu_2$ with transport cost $\kappa(\cdot)$.
    \end{enumerate} 
    Then, there is a set-transport for $(\mu_1,\cost_1)$ of cost $\alpha \cdot \kappa(\eps)$.
Furthermore, if $f,g$ are the mappings used in the $\alpha$-reduction, then we have the following extra properties.
    \begin{enumerate}
        \item   If the set-transport of $(\mu_2,\cost_2)$ is an \emph{oracle} set-transport, so is that of $(\mu_1,\cost_1)$.
        In that case, the complexity of the set-transport algorithm $B$ of $(\mu_2,\cost_2)$ is bounded as follows. If $A$ asks $k$ membership queries and $k'$ sampler queries, then $B$ runs $A$ once while it runs $g$ $k+1$ times and $f$ $k'+1$ times.
        Hence, if $f,A,g$ are efficiently computable, so is $B$.
        \item   If $A$ and the reduction (i.e., $f,g$) are online, then so is $B$. 
    \end{enumerate}
\end{theorem}
\begin{proof}
We directly describe the proof for the \emph{oracle} set membership case.
    Algorithm $B$ is given oracle access to membership queries for a set $\cS_1$, an oracle sampler for $\mu_1$, and an input point $x_1 \in \Supp(\mu_1)$. It then performs as follows.
    \begin{enumerate}
        \item It gets $x_2 \sim f(x_1)$.
        \item It runs $A^{\cS_2,\mu_2}(x_1)$ with respect to the set $\cS_2 = g^{-1}(\cS_1)$ to get $x'_2$. Along the way,
        \begin{itemize}
            \item $B$ provides $A$ with oracle access to $\cS_2$ using $g$ and its own $\cS_1$ membership oracle (leading to $k$ $g$-executions), and
            \item $B$ provides $A$ sampler oracle for $\mu_2$ by using its own $\mu_1$ sampler and applying $f$ to it (which leads to $k'$ $f$-executions).
        \end{itemize} 
        \item Finally, it outputs $x'_1 = g(x'_2)$. 
    \end{enumerate}
    The online aspect of $B$ (if $A$ is so) is straightforward.
\end{proof}

\parag{Deriving Algorithmic and Online Transport for the  Unit Cube.} 
We now use transport results of Gaussian space to derive similar results for the unit cube $(0,1)^n$.

\begin{proposition} \label{prop:Gauss-2-Unit}
    Let $M_1=(\mu,\cost), M_2 = (\nu,\cost)$ in which $\mu$ 
      is the uniform distribution over $(0,1)^n$, 
    and that $\nu$ is the standard Gaussian in dimension $n$. 
    Then there is a $1$-Lipschitz reduction from oracle transports in $M_1$ to oracle transports in $M_2$ for both  $\cost=\ell^p_p$ and $\cost =\ell_p$.
\end{proposition}
\begin{proof}    
    We   define the mappings  $f,g$. They are both deterministic, online, efficiently computable,  and the inverses of each other.
 For $y = (y_1,\dots,y_n) \in \R^n$, $g(y) = (\Phi(y_1),\dots,\Phi(y_n))$ where $\Phi$ is the CDF of the standard normal distribution.  For $x \in (0,1)^n$,  $f(x)=g^{-1}(x)$.
 
        The desired properties follow from the fact that $\Phi$ is $1$-Lipschitz.
\end{proof}

Using Proposition~\ref{prop:Gauss-2-Unit}  and Theorem~\ref{thm:using-reductions}, we can conclude that any (algorithmic/online) transports   for the Gaussian   under the $\ell_p$ or $\ell_p^p$ costs imply similar results for the uniform   over  $(0,1)^n$.

\begin{corollary} \label{cor:unit}
Let $\mu$ be the uniform distribution over $(0,1)$. If $p \geq 1$
then there is an oracle set-transport algorithm $A$ of time-complexity $\poly(n/\eps)$ with cost $\kappa(\eps)$, in which $\eps = \mu(\cS)$ is the measure of the target set $\cS$, for all the cases below.
\begin{enumerate}
    \item If $\cost = \ell_p, p \in [1,2)$, then    $\kappa(\eps) = n^{1/2-1/p} \sqrt{2 \ln \nf{1}{\eps}}$. If $p \geq 2$, then $\kappa(\eps) =   \sqrt{2 \ln \nf{1}{\eps}}$
    \item If $\cost = \ell^p_p, p \geq 2$, then    $\kappa(\eps) = 2 n^{1-2/p}   \ln \nf{1}{\eps} $. If $p \geq 2$, then $\kappa(\eps) =   2    \ln \nf{1}{\eps} $
\end{enumerate}
\end{corollary}
\begin{proof}
    The results for $p \in [1,2]$ follow from Theorem~\ref{thm:Gaussian-One-Way} about Gaussian measure and Proposition~\ref{prop:Gauss-2-Unit}. The results for $p\geq 2$ follow from the fact that $\ell_p(x,y) \leq \ell_q(x,y)$ for $p \leq q$ whenever $x,y \in [0,1]^n$, and hence for $p\geq 2$, we use the same algorithm for $p=2$.
\end{proof}

\parag{Deriving Algorithmic  Transport for the Sphere.}
We now derive algorithmic (but not online) transports for the uniform distribution over the spheres.

\begin{theorem}\label{thm:sphere-2-Gauss}
      Let $M_1=(\mu,\cost), M_2 = (\mu,\ell_2)$, $M_3=(\nu,\ell_2)$ in which,
      \begin{itemize}
          \item $\mu$
      is the uniform distribution over the  sphere $ S_n = \set{x \mid \ell_2(x) = \sqrt n}$;
      \item $x \sim \nu$ is distributed as standard Gaussian in dimension $n$ conditioned on $\ell_2(x) \geq \sqrt n$;
      \item $\cost$ is the   spherical  (aka great-circle) distance.
      \end{itemize}
      Then, there is a $1$-reduction from (resp. oracle) set transports for  $M_1$ to (resp. oracle) set transports in $M_2$, and
      there is a $\pi$-reduction from (resp. oracle) set transports for  $M_2$ to (resp. oracle) set transports in $M_3$.  Hence, by composition,\footnote{See Corollary \ref{cor:Gaussian-Two-Way}.} there is also a $\pi$-reduction from (resp. oracle) set transports for $M_1$ to (resp. oracle) set transports in $M_3$.
\end{theorem}
\begin{proof}
The second reduction relies on a relation between the $\ell_2$ and the spherical distance, stated in the following. For $z,w \in S_n$, $s(z,w) \leq \pi \cdot \ell_2(z,w)$, where $s(z,w)$ is their spherical distance.\footnote{The  equality holds when $z,w$ are the endpoints of a diameter of $S$.}  \Mnote{Is there a citation for this fact?} Therefore, for the second reduction $R'=(f',g')$ we can pick the mappings $f',g'$ as identify functions, and simply rely on the stated Lipschitz property.

We  define the mappings $f,g$ of the first reduction. For both it is easier to define the function $\rho_d(x) = d \cdot x / \ell_2(x) $ for $d \in \R, x\in \R^n$, which rescales $x$ to the length $d$.
\begin{itemize}
\item Given $x \in S_n$, we first sample  $d \sim \ell_2(x)$ for $x \sim \nu$ and then let $f(x) = \rho(x,d)$. 
\item For $x \in \R^n, \ell_2(x) \geq \sqrt n$, $g(x) = \rho(x,\sqrt n)$ is the inverse of $f$ that  projects $x$ to $S_n$. 
\end{itemize}
The fact that $f$ and $g$ transport the distributions the right way follows from the well-known relation between (the symmetry of)  the Gaussian measure and the uniform measure over the sphere. We   analyze the Lipschitz property of $g$.

Let $x,y$ be two points outside the sphere $S_n$, and $x'=g(x),y'=g(y)$.  We claim that $\ell_2(x',y')\leq \ell_2(x,y)$. \Mnote{Is there a citation for this to avoid the proof below?} To prove this, we rely on the following fact: If $w'=\rho(w,r')$ for $r' \geq r$ (i.e., $w'$ is outside $S$ and projects to $w$ on $S$), then $\ell_2(z,w)\leq \ell_2(z,w')$. \Mnote{or a citation for this?}

Now, assume w.l.o.g. that $ \ell_2(y) \geq \ell_2(x)$, and let $w = \rho(y,\ell_2(x))$. We then have:
$$ \ell_2(x,y) \geq \ell_2(x,w) \geq \ell_2(x',y'),$$
where the first inequality follows from the stated fact above, and the second one is because $x',y'$ are contractions of $x,w$ to a smaller radius.
\end{proof}

We now use the reductions above and the (algorithmic) transportation for the standard Gaussian to derive (algorithmic) transportation   about the uniform measure over the sphere.
\begin{theorem} \label{thm:sphere}
    Let $\mu$ be the uniform measure on the sphere $S_n$ of radius $\sqrt n$. For $\kappa(\eps) = \sqrt{2 \ln \nf{1}{\eps}}$,  
    \begin{itemize}
        \item $(\mu,\ell_2)$ has an oracle transport of cost $\kappa(\cdot)+1.52$.
        \item $(\mu,s)$, where $s$ is the spherical distance, has an oracle transport of cost $\pi \cdot (\kappa(\cdot)+1.52).$
    \end{itemize}
\end{theorem}
\begin{proof}
    Let $\Phi^n$ be the standard Gaussian measure in $\R^n$. Consider the set $\cE$ of points outside $S_n$, and let $\mu_\cE$ be the $\mu$ conditioned on $\cE$. 
    
    It holds that $\lim_{n \to \infty} \mu(\cS) =1/2$ and for all $n$ we have $\mu(\cS) \approx 0.3174$.
    Using the reversibility of our results and the composition,\footnote{See Corollary \ref{cor:Gaussian-Two-Way} for an application of composition.}  we can transport $\Phi$ to any    $\Phi|\cS$ for $\cS$ of measure $\Phi(\cS)=\eps$  by going through $\cN(0,1)$  with expected cost at most  $\kappa(\eps)+\kappa(0.3173) \leq \kappa(\eps) + 1.52$. The cost of this composition can be bound using the triangle inequality of Proposition \ref{prop:triangle}.
    Therefore, the claims of the theorem follow from the above bound and the reductions of Theorem~\ref{thm:sphere-2-Gauss}.
\end{proof}

%% file: onedim.tex
\subsection{Transport in One Dimension}
\parag{Notation.} For a distribution $\mu$, we use $F_\mu$ to denote its Cumulative Distribution Function (CDF for short). $F_\mu$ is right-continuous, and $x$ is called an atom if $F$ is not continuous at $x$. For a real function $F$, $F^{-1}$ denotes its ``generalized inverse'' defined as $F^{-1}(t) = \inf \set{x \in \R \mid F(x)>t}$.  For $\cX=(x_1,\dots,x_k)$, we let $U_{\cX}$ be the uniform distribution over the multi-set $\set{x_1,\dots,x_n}$. If $x_i \sim \mu$ for a distribution $\mu$ for all $i \in [k]$, we sometimes refer to $U_\cX$ as an empirical distribution (related to $\mu$) of size $k$.


\begin{theorem}[Optimal Transport in $\R$] \label{thm:opt1D}
    Let $\mu,\nu$ be two distributions over $\R$ with CDFs $F_\mu,F_\nu$, and let $\cost$ be a convex transportation cost function. Then the optimal transport can be obtained using the unique monotone transport from $\mu$ to $\nu$. In particular, the  two-dimensional distribution $\pi$ with the following density function gives an optimal transport between $\mu,\nu$:
        $$F_\pi(x,y) = \min \set{F_\mu(x),F_\nu(y)}.$$
\end{theorem}


The following lemma gives an explicit simple algorithm to compute the  transport.
It can be obtained from  Theorem 2.18   in \cite{villani2021topics}.\footnote{See the proof of Remark 2.19 in \cite{villani2021topics} for the details.} 

\begin{lemma}[Algorithmic Optimal Transport in $\R$] \label{lem:AlgOpt1D}
Let $\mu,\nu$ be two distributions over $\R$ with CDFs $F_\mu,F_\nu$, and let $\cost$ be a convex function. Consider the following algorithm $A(\cdot)$ that   produces output $y$ given an input $x$.
    \begin{enumerate}
        \item Let $p_0 = \lim_{z \To x^{-}} F_\mu(z)$ be the left limit, and $    p_1 = F_\mu(x) \geq F_\mu(z)$.
         \item Sample $t \in [p_0,p_1]$ uniformly.
         \item Output $F_{\nu}^{-1}(t)$.
    \end{enumerate}
Then, for $x \sim \mu$, it holds that $A(x) \sim \nu$, and the coupling achieves the optimal transport.
\end{lemma}

 \begin{remark}[Special Cases of Lemma~\ref{lem:AlgOpt1D}] \label{rem:specials}
The following   special cases hold for Lemma~\ref{lem:AlgOpt1D}.
\begin{itemize}
    \item When   $F_\mu$ is continuous (i.e., no atom points in $\mu$), then $A(x)=F^{-1}_\nu(F_\mu(x))$.
    \item When $\mu=U_\cX,\nu=U_\cY$ for $\cX = (x_1 \leq \dots \leq x_k), \cY = (y_1 \leq \dots \leq y_k)$, then $A(x_i)=y_i$.
\end{itemize}
\end{remark}

%% file: transport-entropy.tex
\subsection{Transport-Entropy Inequalities} \label{sec:Trans-Ent}
In this subsection, we show how to borrow transportation-entropy inequalities from the literature and interpret them in a way that is useful for our context.

\begin{lemma}[Transport-Entropy Inequalities for the Normal Distribution] \label{lem:TE-one-dim}
  If $\cost(x,y) = |x-y|$, $p \in [1,2]$,  $\mu = \cN(0,1)$,   and $\nu$ is an arbitrary distribution on $\R$, then 
  $$\T^{2/p}_{\cost^p}(\mu,\nu) \leq  2 \KL{\nu}{\mu} .$$
  \remove{
   If $\mu$ is uniform over $[+1,-1]$ and $\Supp(\nu) \leq [-1,+1]$, then
  $$2 \T^2_{p,\cost}(\mu,\nu) \leq      \KL{\nu}{\mu} .$$ }
 \end{lemma}
\begin{proof}
 For $p=2$, this  is the same as $\T_{\cost^2}(\mu,\nu) \leq  2 \KL{\nu}{\mu}$, which is  known as Talagrand's $\mathbf{T}_2$ inequality for $\cN(0,1)$ \cite{talagrand1996transportation}. For other $p \in [1,2]$, the inequality follows from the monotonicity  of Wasserstein distance stating that  for every coupling $\pi$ and $1 \leq q \leq p$, by Jensen's inequality we have $\T^{1/q}_{\cost^q}(\pi) \leq \T^{1/p}_{\cost^p}(\pi)$. 
 \remove{
    Consider $t(x) = 2 \Phi(x)-1 $, where $\Phi(\cdot)$ is the CDF of the standard normal distribution, and let $t(\omega)$ be the distribution of the random variable $t(x)$ where $ x \sim \omega$. Importantly, if $t = \cN(0,1)$, then  $t(\omega)$ is the uniform distribution over $(-1,+1)$. Similarly, define the mappings $t^{-1}(y) = \Phi^{-1}((y+1)/2)$ for $y \in (-1,+1)$, and $t^{-1}(\mu)$ for $\Supp(\mu) \se (-1,+1)$.

The lemma follows from the following two claims.
    \begin{enumerate}
        \item  Contracting distances: $\T_{p,\cost}(t(\mu),t(\nu)) \leq \T_{p,\cost}(\mu,\nu)/2$. We prove a more general statement that   $|t(y) - t(x)| < |y-x|/2.$ Then, by the monotonicity of $f(x) = x^p $ for $x\geq 0$ (assuming $p\geq 0$), we get $\T_{\cost^p}(t(\mu),t(\nu)) \leq \T_{\cost^p}(\mu,\nu) \cdot 2^{-p}$, which implies the claim.
        The reason   is that the   maximum of the absolute value of the derivative of $\cN(0,1)$ is achieved at $x=\pm 1$, and that derivative is less than $0.242 < 1/4$.

        \item Preserving divergence: For distributions $\mu,\nu$, it holds that $\KL{\nu}{\mu} = \KL{t(\nu)}{t(\mu)}$. \Mnote{It is intuitively clear as KL does not depend on where the probability masses are, but what is the formal proof here?}
    \end{enumerate}
    Putting things together,
    $$2 \T^2_{p,\cost}(\mu,\nu) \leq 2(\T^2_{p,\cost}(t^{-1}(\mu),t^{-1}(\nu))/2)^2 \leq \KL{t^{-1}(\nu)}{t^{-1}(\mu)} = \KL{\nu}{\mu},$$
    where the last inequality follows from the fact that $t^{-1}(\mu) = \cN(0,1)$ and the first part of the lemma about $\mu = \cN(0,1)$.
}\end{proof}

The following lemma is essentially due to   \cite{gozlan2007large,gozlan2010transport}, but we re-formulate it to shed light on the key quantity $\Delta_{\cost}$ that is implicit in their proof.
\begin{lemma}[Tensorization of Transport-Entropy Inequalities \cite{gozlan2007large,gozlan2010transport}] \label{lem:tensor}
    Suppose:
    \begin{enumerate}
        \item The function $\alpha \colon [0,\infty) \To [0,\infty)$ is convex and increasing.
        \item For $i \in [n]$ for distribution $\mu_i$ and every distribution $\sigma$ we have $\alpha(\T_{\cost_i}(\mu_i,\sigma)) \leq \KL{\sigma}{\mu_i}$.
        \item Transport cost $\cost$ is  linear over $\cost_1,\dots,\cost_n$; namely, $\cost(\ol{x},\ol{y}) = \sum_{i \in [n]} \cost_i(x_i,y_i)$.
    \end{enumerate}
      Then, for $\src = \mu_1 \otimes \dots \otimes \mu_n$ and arbitrary $\trg$  we have
      $$  \alpha\left(\frac{\lowerprod_{\cost}(\mu,\nu)}{n}\right) \leq \frac{\KL{\trg}{\src}}{n}.$$
\end{lemma}
\begin{proof}
The lemma  is basically Proposition 1.9 in \cite{gozlan2010transport} with the only difference that there the result is stated for $\T_{\cost}(\src,\trg)$ rather than $\lowerprod_{\cost}(\mu,\nu)$, but an inspection of the proof shows that it indeed holds in the stronger form   stated here, which bounds optimal \emph{online} (rather than   offline) couplings.
\end{proof}


\begin{theorem}[Talagrand's Transportation Inequality for Gaussian Measure] \label{thm:Talagrand}
  For   $\Phi_n$ is the standard Gaussian in dimension $n$,  and $\nu$ is an arbitrary distribution on $\R^n$,
  $$\T^2_{\ell_2}(\Phi_n,\nu) \leq  {2  \KL{\nu}{\Phi_n}} .$$
\remove{
If $\mu$ is uniform over $[+1,-1]^n$ and $\Supp(\nu) \leq [-1,+1]^n$,  
  $$ \lowerprod_{p,\cost}(\mu,\nu) \leq   n^{1/2-1/p}  \cdot \sqrt { \KL{\nu}{\mu}/2} .$$
} \end{theorem}

The following corollary follows from Lemmas~\ref{lem:TE-one-dim} and~\ref{lem:tensor} and generalizes Theorem \ref{thm:Talagrand}.

\begin{corollary}[Transport-Entropy   Inequalities for Gaussian Measure] \label{cor:genTal}
  If $\cost(x,y) = \ell^p_p(x,y)$, $p \in [1,2]$, $\Phi_n$ is the standard Gaussian  and $\nu$ is an arbitrary distribution both in $\R^n$, then  
  $$\T^\OnT_{\cost}(\Phi_n,\nu) = \lowerprod_{\cost}(\Phi_n,\nu) \leq n^{1-p/2} \cdot (2\KL{\nu}{\Phi_n})^{p/2} .$$
\end{corollary}

As a special case, when $p=2$, the above corollary gives us the Talagrand's transportation inequality for the Gaussian measure in dimension $n$.

%% file: empirical.tex
 \subsection{Cost of Transport to Empirical  Gaussian}
%



The following result from \cite{fournier2015rate} states   conditions  for   upper bounding the empirical to original transport based on the moments. Recall Definition \ref{def:emp} for the notation $\empT_{\cost,k}(\mu)$.
\begin{lemma}[Bounding the Original-to-Empirical Transport  Using Moments \cite{fournier2015rate}] \label{lem:moments}  Let $\cost$ be $\ell_p^p$  for $p\geq 1$ (i.e., $\cost(x,y) = \sum_{i \in [n]} |x_i-y_i|^p$). If $M_q(\mu) := \Ex_{x \sim \mu} |x|^q <\infty$ for $q>2p$, then 
$$\empT_{\cost,k}(\mu) \leq C_p \cdot 2^p \cdot D_{q/2-p} \cdot  M_q^{p/q}\left(\mu\right) \cdot k^{-1/2},$$
for $D_s = \frac{1}{1-2^{-s}}$ and decreasing $C_p  = 1+o_p(1)$ for $p\geq 1$.
\end{lemma}
\begin{proof}
Lemma~\ref{lem:moments} follows from   Theorem~1 in \cite{fournier2015rate} by putting dimension $d=1$, limiting  to $p\geq 1$, and explicitly finding a constant $C$ from the proof of Theorem~1 in \cite{fournier2015rate}. Another point is that by tracking the proof (in Step 2 of their proof) one realizes that for $q>2p$ we can only use $k^{-1/2}$ in the right hand side (i.e., the upper bound) rather than $k^{-1/2} + k^{-(q-p)/q}$.

Our constant $C_p \cdot 2^p \cdot D_{q/2-p}$   is simply an upper bound on $C$ from Theorem~1 of \cite{fournier2015rate}. In the following, we list the constants (and their corresponding part in the proof) whose multiplication  results in the constant $C$. (All the references   are to the content of \cite{fournier2015rate}). The names $D_{s}$ are chosen by us, and it is a short form for   $D_s = 1/(1-2^{-s})$. The constants are as follows: $\kappa_{p} =2^p (2^p+1)/(2^p-1) $ (Lemma 5), $D_p$ (Lemma 6), $D_{p-1/2}$ (Step 1), and $D_{q/2-p}$ (Step 2). 
The multiplication 
$\kappa_{p}   D_p   D_{p-1/2}$ is equal to $ 2^p \cdot  C_p$ for
$$C_p =  \frac{(2^p+1)}{(2^p-1)(1-2^{-p})(1-2^{-p+1/2})}.$$
Finally,   $C_p$ is decreasing for $p\geq 1$ with maximum $12+5\sqrt{2}<21$ achieved    at $p=1$. 
\end{proof}

\remove{
\begin{definition}
\label{def:dom}
For constant $\eps \in (0,1]$, we say that a a distribution $\mu$ \emph{$\eps$-dominates} distribution $\nu$, if for all events $\cT$ we have $\nu(\cT) \leq   \mu(\cT)/\eps$.\footnote{Note that the larger $\eps$, the better the domination is. In particular, for $\eps\approx 0$ we get nothing and for $\eps \approx 1$, we have $\mu \approx \nu$.}
\end{definition}
For example,  $\nu$  $\eps$-dominates $\mu$ if $\nu$ is $\mu$  conditioned on an event of probability at least $\eps$ in $\mu$.
 }
 
The following corollary follows from Lemma~\ref{lem:moments}.
\begin{corollary}[Original-to-Empirical Transport for the Standard Normal Distribution] \label{cor:empGaussian} Let $p \geq 1$,   $\cost$ be $\ell_p^p$, $C_p$ be the constant of Lemma~\ref{lem:moments}, and $\mu =\cN(0,1)$ is the standard normal distribution. Then 
        $$\empT_{\cost,k}(\mu) \leq  C_p  \cdot 2^{1+3p/2} \cdot \Gamma(p+1)^{\frac{p}{2p+1}} \cdot k^{-1/2}.$$
\remove{
 If $\Supp(\nu) \se [0,1]$, then  
                $$\empT_{\cost,k}(\nu) \leq  C_p   \cdot k^{-1/2}.$$
                In particular, for $p=1,2$, we (in order) get the upper bounds $21 k^{-1/2}$  and  $3.5 k^{-1/2}$. 
}\end{corollary}
\begin{proof}
   We   rely on the  fact  that when $\mu=\cN(0,1)$,
     $$ M_q(\mu) = \frac{2^{q/2}}{\sqrt \pi} \Gamma\left(\frac{q+1}{2}\right),$$
     in which $\Gamma(\cdot)$ is the gamma function, and we   pick $q=2p+1$.
\remove{
    The the second part, we work with $[-1/2,1/2]$ instead, which is equivalent to $[0,1]$ for us. Then, we have   $M_q(\mu) \leq 2^{-q}$ for all $q \geq 1$, and the claim follows because that $D_{q/2-p} \To 1$ as $q \To \infty$.} 
\end{proof}
 